%% file: paper.tex
\newif\ifanon\anonfalse
\newif\iffull\fulltrue
\documentclass[]{llncs}

\usepackage{ltlfonts}
\usepackage{stmaryrd}
\usepackage{paralist}

\input{inc}

\input{macros}

\begin{document}

\mainmatter
\pagestyle{headings}

\title{Temporal Logics for Hyperproperties}

\ifanon
\else
\author{Michael R. Clarkson$^1$, Bernd Finkbeiner$^2$, Masoud Koleini$^1$, \\
  Kristopher K.\ Micinski$^3$, Markus N. Rabe$^2$, and C\'esar S\'anchez$^4$}
\institute{$^1$George Washington University; $^2$Universit\"at des Saarlandes; \\
$^3$University of Maryland, College Park; $^4$IMDEA Software Institute}
\fi

\maketitle

\begin{abstract}
  Two new logics for verification of hyperproperties are proposed.
  Hyperproperties characterize security policies, such as
  noninterference, as a property of sets of computation
  paths. Standard temporal logics such as LTL, CTL, and \ctlstr{} can refer
  only to a single path at a time, hence cannot express many
  hyperproperties of interest. The logics proposed here, HyperLTL and
  \Hpctl{}, add explicit and simultaneous quantification over multiple paths to LTL
  and to \ctlstr{}. This kind of quantification enables expression 
  of hyperproperties.  A model checking algorithm for the proposed logics is given. 
  For a fragment of HyperLTL, a prototype model checker has been implemented. 
\end{abstract}


\input{intro}

\input{hyperltl}

\input{security}

\input{hyperctl}

\input{relatedlogics}

\input{modelcheckingsat}

\input{prototype}

\input{relatedwork}


\input{conclusion}

\ifanon
\relax
\else
 \paragraph*{{\bf Acknowledgements.}}
 Fred B. Schneider suggested the name ``HyperLTL.''
 We thank him, Rance Cleaveland, Rayna Dimitrova, Dexter Kozen, Jos\'e Meseguer, and Moshe Vardi for discussions about this work.
 Adam Hinz worked on an early prototype of the model checker.
 This work was supported in part by AFOSR grant FA9550-12-1-0334, NSF grant CNS-1064997, the German Research Foundation (DFG)
 under the project SpAGAT within the Priority Program 1496 ``Reliably
 Secure Software Systems --- RS$^3$,'' and Spanish Project ``TIN2012-39391-C04-01 STRONGSOFT.''
\fi

\bibliographystyle{abbrv}
\bibliography{References_short,bib_short}

\appendix

\iffull
\input{extendedEncodings}

\input{asynch}

\input{satisfiabilityAppendix}
\input{constructions}

\input{prototypeDetails}
\else
\relax
\fi

\end{document}

%% file: inc.tex
\usepackage{amsmath}
\usepackage{amssymb}

\usepackage{amsthm}
\usepackage{amscd}
\usepackage{amsfonts}
\usepackage{graphicx}
\usepackage{xcolor}
\usepackage{fancyhdr}
\usepackage{url}
\usepackage{comment}
\usepackage{hyperref}
\usepackage{tikz}
\usepackage{setspace}
\usepackage{cite}

\usetikzlibrary{automata,positioning}

\newcommand{\A}{\ensuremath{\mathop{\text{A}}}}

\newcommand{\U}{\ensuremath{\mathbin{\text{U}}}}
\newcommand{\R}{\ensuremath{\mathbin{\text{R}}}}

\newcommand{\AAx}{\ensuremath{\mathop{\text{AA}}}}

\newcommand{\X}{\ensuremath{\mathop{\text{X}}}}
\newcommand{\G}{\ensuremath{\mathop{\text{G}}}}
\newcommand{\F}{\ensuremath{\mathop{\text{F}}}}

\renewcommand{\iff}{\ensuremath{\mathop{\text{iff}}}}

\newcommand{\cl}{\mathit{cl}}
\newcommand{\ms}{\mathit{ms}}
\newcommand{\Buchi}{B\"uchi}
\newcommand{\proj}{\mathit{prj}}

\newcommand{\zip}{\ensuremath{\mathit{zip}}}
\newcommand{\unzip}{\ensuremath{\mathit{unzip}}}

\newcommand{\Hp}{$\text{HyperLTL}_2$}
\newcommand{\Hpctl}{$\text{HyperCTL}^*$}
\newcommand{\ctlstr}{$\text{CTL}^*$}

\newcommand{\IF}{\mbox{\bf if}}
\newcommand{\THEN}{\mbox{\bf then}}
\newcommand{\ELSE}{\mbox{\bf else}}

\newcommand{\WHILE}{\mbox{\bf while}}
\newcommand{\plDO}{\mbox{\bf do}}
\newcommand{\cond}[3]{\ensuremath{\IF ~ #1 ~\THEN~ #2 ~\ELSE~ #3}}
\newcommand{\while}[2]{\ensuremath{\WHILE \: #1 \:\plDO\: #2}}

\DeclareMathSymbol{\Pi}{\mathalpha}{operators}{5}
\DeclareMathSymbol{\Phi}{\mathalpha}{operators}{8}

\renewcommand{\LTLdiamond}{\ensuremath{\F}}
\renewcommand{\LTLcircle}{\ensuremath{\X}}
\renewcommand{\LTLsquare}{\ensuremath{\G}}

\newcommand{\seeappendix}[1]{\iffull{#1}\else{the companion technical report~\cite{hypertl-tr}}\fi}

\usepackage{hp_defs}

%% file: macros.tex
\newcommand{\donotshow}[1]{}

\newcommand{\Out}{\ensuremath{\mc{V}_\mc{O}} }
\newcommand{\out}{\ensuremath{\mathsf{out}} }

\newcommand{\AltPaths}{\mathsf{AltPaths}}

\newcommand{\false}{\mathsf{false}}

\newcommand{\true}{\mathsf{true}}

\newcommand{\V}{\ensuremath{\mathcal{V}} }

\newcommand{\W}{\ensuremath{\mathop{\text{W}}}}
\renewcommand{\P}{\ensuremath{\mathsf{P}} }
\newcommand{\AP}{\ensuremath{\mathsf{AP}} }
\newcommand{\TR}{\ensuremath{\mathsf{TR}} }

\newcommand{\Hide}{\ensuremath{\mathcal{H}\hspace{1pt}} }

\renewcommand{\do}{\mathsf{do}}
\newcommand{\DO}{\mathsf{do}}
\renewcommand{\Out}{\mathit{Out}}
\newcommand{\Input}{\mathit{In}}

\newcommand{\Paths}{\ensuremath{\mathsf{Paths}} }
\newcommand{\Traces}{\ensuremath{\mathsf{Traces}} }

\newcommand{\Agts}{\ensuremath{\mathsf{Agts}}}



%% file: intro.tex

\section{Introduction}
\label{sec:intro}

\emph{Trace properties}, which developed out of an interest in proving the correctness of
programs~\cite{Lamport77}, characterize correct behavior as 
properties of individual execution traces.
Although early verification techniques specialized in proving
individual correctness properties of interest, such as mutual
exclusion or termination, temporal logics soon emerged as a general,
unifying framework for expressing and verifying trace properties.
Practical model checking
tools~\cite{Holzmann:1997:SPIN,NuSMV:2000:Cimatti,Cook+Koskinen+Vardi/2011/TemporalPropertyVerificationAsAProgramAnalysisTask}
based on those logics now enable automated verification of program
correctness.

Verification of security is not directly possible with
such tools, because some important security policies cannot be characterized
as properties of individual execution traces~\cite{McLean:1994:GeneralTheory}. 
Rather, they are properties of sets of execution traces, also known as 
\emph{hyperproperties}~\cite{ClarksonS10}.
Specialized verification techniques have been developed for
particular hyperproperties~\cite{HammerS09,BanerjeeN05,Myers99,MilushevC13},
as well as for \emph{2-safety} properties~\cite{Terauchi+Aiken/05/SecureInformationFlowAsSafetyProblem}, which 
are properties of pairs of execution traces. 
But a unifying program logic for expressing and verifying
hyperproperties could enable automated verification of a
wide range of security policies.

In this paper, we propose two such logics. Both are based, like
hyperproperties, on examining more than one execution trace at a time.
Our first logic, \emph{HyperLTL}, generalizes linear-time temporal
logic (LTL)~\cite{Pnueli:1977:TLP}.  LTL implicitly quantifies over
only a single execution trace of a system, but HyperLTL allows
explicit quantification over multiple execution traces simultaneously,
as well as propositions that stipulate relationships among those
traces.  For example, HyperLTL can express information-flow
policies such as \emph{observational 
determinism}~\cite{McLean92,Roscoe95,ZdancewicM03}, which requires
programs to behave as (deterministic) functions from low-security
inputs to low-security outputs.  The following two programs do not
satisfy observational determinism, because they leak the value of
high-security variable $h$ through low-security variable $l$, thus making
the program behave nondeterministically from a low-security user's
perspective:
\begin{equation*}
\begin{array}{l@{\hspace{4em}}l}
\text{(1)}~~l := h & \text{(2)}~~\cond{h=0}{l:=1}{l:=0}
\end{array}
\end{equation*}
Other program logics could already express observational determinism or closely related policies~\cite{BartheDR04,MilushevC13,HuismanWS/06/TLCharacterisationOfOD}. 
Milushev and Clarke~\cite{MilushevC12,Milushev:2013:thesis,MilushevC13} have even
proposed other logics for hyperproperties, which we discuss in Section~\ref{sec:relatedwork}.
But HyperLTL provides a simple and unifying logic in which many information-flow security policies
can be directly expressed.

Information-flow policies are not one-size-fits-all.  Different policies might be needed 
depending on the power of the adversary. For example, the following
program does not satisfy observational determinism, but the program might be
acceptable if nondeterministic choices, denoted $\talloblong$, are
resolved such that the probability distribution on output value $l$ is
uniform:
\begin{equation*}
\text{(3)} ~~ l := h ~\talloblong~ l := 0 ~\talloblong~ l := 1
\end{equation*}
On the other hand, if the adversary can influence the resolution of nondeterministic choices, program (3) could be exploited to leak information.
Similarly, the following program does satisfy observational
determinism, but the program might be unacceptable if adversaries 
can monitor execution time:
\begin{equation*}
\text{(4)} ~~ \while{h > 0}{\{h := h-1\}}
\end{equation*}%
\setcounter{equation}{4}
In Section~\ref{sec:examples}, we show how policies appropriate for the above 
programs, as well as other security policies, can be formalized in HyperLTL.


Our second logic, \emph{\Hpctl{}}, generalizes a branching-time
temporal logic,
\ctlstr{}~\cite{EmersonH86}.
Although \ctlstr{} already has explicit trace quantifiers,
only one trace is ever in scope at a given point in a formula (see
Section~\ref{sec:ltlctl}),
so \ctlstr{} cannot directly express hyperproperties.
But \Hpctl{} can, because it permits quantification over multiple 
execution traces simultaneously.
HyperLTL and \Hpctl{} enjoy a similar relationship to that of LTL and \ctlstr{}:
HyperLTL is the syntactic fragment of \Hpctl{} containing only formulas in \emph{prenex} form---that is, formulas that begin exclusively with quantifiers and end with a quantifier-free formula.
\Hpctl{} is thus a strict generalization of HyperLTL.
\Hpctl{} also generalizes a related temporal logic, SecLTL~\cite{Dimitrova:2012:SecLTL}, and subsumes epistemic temporal logic~\cite{Fagin-book-95,Meyden/1993/AxiomsForKnowledgeAndTimeInDistributedSystemsWithPerfectRecall} (see Section~\ref{sec:relatedlogics}). 

Having defined logics for hyperproperties, we investigate model checking of those logics. 
In Section~\ref{sec:ctlmodelchecking}, we show that for \Hpctl{} the 
model checking problem is decidable by reducing it to the 
satisfiability problem for quantified propositional temporal 
logic (QPTL)~\cite{Sistla+Vardi+Wolper/1987/TheComplementationProblemForBuchiAutomata}. 
Since \Hpctl{} generalizes HyperLTL, we immediately obtain that the HyperLTL model checking
problem is also decidable. 
We present a hierarchy of fragments, which allows us to precisely characterize the complexity of the model checking problem in the number quantifier alternations. 
The lowest fragment, which disallows any quantifier alternation, can be checked by a space-efficient polynomial-time algorithm (NLOGSPACE in the number of states of the program). 

We also prototype a model
checker that can handle an important fragment of HyperLTL, including all the
examples from Section~\ref{sec:examples}.  The prototype implements a
new model checking algorithm 
based on a well-known LTL algorithm~\cite{Vardi:1994:infComp,Wolper00} and on
a self-composition construction~\cite{BartheDR04,Terauchi+Aiken/05/SecureInformationFlowAsSafetyProblem}. 
The complexity of our algorithm is exponential in the size of
the program and doubly exponential in the size of the
formula---impractical for real-world programs, but at least a
demonstration that model checking of hyperproperties formulated in our logic is possible.

This paper contributes to theoretical and foundational aspects of security by:
\begin{compactitem}
\item defining two new program logics for expressing hyperproperties,
\item demonstrating that those logics are expressive enough to formulate important in\-for\-ma\-tion-flow policies, 
\item proving that the model checking problem is decidable, and
\item prototyping a new model checking algorithm and using it to verify security policies.
\end{compactitem}

The rest of the paper is structured as follows.  Section
\ref{sec:hyperltl} defines the syntax and semantics of HyperLTL.
Section~\ref{sec:examples} provides several example formulations of
information-flow policies.
Section~\ref{sec:hyperctl} defines the syntax and semantics of
\Hpctl{}.  Section~\ref{sec:relatedlogics} compares our two logics
with other temporal and epistemic logics.
Section~\ref{sec:ctlmodelchecking} obtains a model checking algorithm
for \Hpctl{}.  Section~\ref{sec:prototype} describes our prototype
model checker.  Section~\ref{sec:relatedwork} reviews related work,
and Section~\ref{sec:conclusion} concludes.

%% file: hyperltl.tex
\section{HyperLTL}
\label{sec:hyperltl}

HyperLTL extends propositional linear-time temporal logic (LTL) \cite{Pnueli:1977:TLP} with explicit quantification over \emph{traces}.
A \emph{trace} is an infinite sequence of sets of \emph{atomic propositions}.
Let $\AP$ denote the set of all atomic propositions.
The set $\TR$ of all traces is therefore $(2^\AP)^{\omega}$.

We first define some notation for manipulating traces.
Let $t\in\TR$ be a trace. 
We use $t[i]$ to denote element $i$ of $t$, where $i \in \mathbb{N}$.
Hence, $t[0]$ is the first element of $t$.
We write $t[0,i]$ to denote the prefix of $t$ up 
to and including element $i$, and $t[i,\infty]$ to denote the infinite suffix of $t$ beginning with element $i$.  

\paragraph*{Syntax.} 
Let $\pi$ be a \emph{trace variable} from an infinite supply $\V$ of trace variables. 
Formulas of HyperLTL are defined by the following grammar:
\[
\begin{array}{lll}
\psi ~ ::= ~~ \exists \pi.\; \psi  &\vert \quad \forall \pi.\; \psi  &\vert \quad \varphi\\
\varphi ~ ::= ~~ a_{\pi} \quad &\vert \quad \neg\varphi \quad &\vert \quad 
\varphi\vee\varphi
\quad \vert \quad~ \LTLcircle\varphi \quad \vert \quad \varphi \U\varphi 
\end{array}
\]
\noindent 
Connectives $\exists$ and $\forall$ are universal and existential trace quantifiers, read as ``along some traces'' and ``along all traces.''
For example, $\forall \pi_1.\; \forall \pi_2.\; \exists \pi_3.\; \psi$ means that for all traces $\pi_1$ and $\pi_2$, there exists another trace $\pi_3$, such that $\psi$ holds on those three traces.
(Since branching-time logics also have explicit path quantifiers, it is natural to wonder why one of them does not suffice to formulate hyperproperties. Section~\ref{sec:ltlctl} addresses that question.)
A HyperLTL formula is \emph{closed} if all occurrences of trace variables are bound by a trace quantifier.

An atomic proposition $a$, where $a \in \AP$, expresses some fact about states.
Since formulas may refer to multiple traces, we need to disambiguate which trace the proposition refers to. 
So we annotate each occurrence of an atomic proposition with a trace variable $\pi$.  
Boolean connectives $\neg$ and $\vee$ have the usual classical meanings.
Implication, conjunction, and bi-implication are defined as syntactic sugar: $\varphi_1\rightarrow\varphi_2 \equiv \neg\varphi_1\vee\varphi_2$, and $\varphi_1\wedge\varphi_2 \equiv \neg(\neg\varphi_1\vee\neg\varphi_2)$, and $\varphi_1 \leftrightarrow \varphi_2 \equiv (\varphi_1 \rightarrow \varphi_2) \wedge (\varphi_2 \rightarrow \varphi_1)$.
True and false, written $\true$ and $\false$, are defined as $a_\pi\vee\neg a_\pi$ and $\neg\true$.

Temporal connective $\LTLcircle\varphi$ means that $\varphi$ holds on the next state of every quantified trace.
Likewise, $\varphi_1\U\varphi_2$ means that $\varphi_2$ will eventually hold of the states of all quantified traces that appear at the same index, and until then $\varphi_1$ holds.
The other standard temporal connectives are defined as syntactic sugar:
$\LTLdiamond\varphi\equiv\true\U\varphi$, and
$\LTLsquare\varphi\equiv\neg\LTLdiamond\neg\varphi$, and
$\varphi_1\W\varphi_2\equiv(\varphi_1\U\varphi_2)\vee\LTLsquare\varphi_1$, 
and $\varphi_1\R\varphi_2 \equiv \neg(\neg\varphi_1\U\neg\varphi_2)$.

We also introduce syntactic sugar for comparing traces.  
Given a set $P$ of atomic propositions, $\pi[0]\!=_P\!\pi'[0] \equiv\bigwedge_{a\in P} a_{\pi}\!\leftrightarrow\! a_{\pi'}$.  
That is, $\pi[0]\!=_P\!\pi'[0]$ holds whenever the first state in both $\pi$ and $\pi'$ agree on all the propositions in $P$. 
And $\pi\!=_P\!\pi' \equiv \LTLsquare (\pi[0]\!=_P\!\pi'[0])$, that is, all the positions of $\pi$ and $\pi'$ agree on $P$.  
The analogous definitions hold for $\neq$.

\paragraph*{Semantics.}
The validity judgment for HyperLTL formulas is written $\Pi \models_T \psi$, where $T$ is a set of traces, and $\Pi:\mathcal{V}\to \TR$ is a \emph{trace assignment} (i.e., a \emph{valuation}), which is a partial function mapping trace variables to traces.
Let $\Pi[\pi\mapsto t]$ denote the same function as $\Pi$, except that $\pi$ is mapped to $t$.
We write trace set $T$ as a subscript on $\models$, because $T$ propagates unchanged through the semantics; we omit $T$ when it is clear from context.
Validity is defined as follows:
\[
\begin{array}{l@{\hspace{1em}}c@{\hspace{1em}}l}
   \Pi\models_T\exists\pi.\; \psi & \text{iff} & \text{there exists $t \in T : \Pi[\pi \mapsto t] \models_T\psi$} \\
   \Pi\models_T\forall\pi.\; \psi & \text{iff} & \text{for all $t \in T : \Pi[\pi \mapsto t] \models_T\psi$} \\
  \Pi\models_T a_{\pi} & \text{iff} & a\in \Pi(\pi)[0]~ \\
  \Pi\models_T \neg \varphi & \text{iff} & \Pi\not\models_T\varphi~ \\
  \Pi\models_T \varphi_1 \vee \varphi_2 & \text{iff} & \Pi\models_T\varphi_1 \text{ or } \Pi\models_T\varphi_2~ \\
  \Pi\models_T\LTLcircle\varphi & \text{iff} & \Pi[1,\infty]\models_{T}\varphi~ \\
  \Pi\models_T\varphi_1\U\varphi_2 & \text{iff} & \text{there exists $i \geq 0 :$ } \Pi[i,\infty]\models_{T}\varphi_2 \\
  && \quad\text{and for all $0 \leq j < i$ we have $\Pi[j,\infty]\models_{T}\varphi_1$}
\end{array}
\]
Trace assignment \emph{suffix} $\Pi[i,\infty]$ denotes the trace assignment $\Pi'(\pi)=\Pi(\pi)[i,\infty]$ for all~$\pi$. 
If $\Pi \models_T \varphi$ holds for the empty assignment $\Pi$, then $T$ \emph{satisfies} $\varphi$.

We are interested in whether programs satisfy formulas, so we first
derive a set $T$ of traces from a program, first using \emph{Kripke
  structures} as a unified representation of programs.  A Kripke
structure $K$ is a tuple $(S,s_0,\delta,\AP,L)$ comprising a set of
\emph{states} $S$, an \emph{initial state} $s_0\in S$, a transition
function $\delta:S\to 2^{S}$, a set of \emph{atomic propositions}
$\AP$, and a labeling function $L:S\to 2^{\AP}$.  To ensure that all
traces are infinite, we require that $\delta(s)$ is nonempty for every
state $s$.


The set $\Traces(K)$ of traces of $K$ is the set of all sequences of labels produced by the state transitions of $K$ starting from initial state. Formally, $\Traces(K)$ contains trace $t$ iff  there exists a sequence $s_0 s_1 \ldots$ of states, such that $s_0$ is the initial state, and for all $i \geq 0$, it holds that $s_{i+1}\in\delta(s_{i})$; and $t[i] = L(s_i)$.
A Kripke structure $K$ \emph{satisfies} $\varphi$, denoted by $K\models \varphi$, if ${\Traces(K)}$ satisfies $\varphi$.

It will later be technically convenient to consider enlarging the set
$\AP$ of atomic propositions permitted by a Kripke structure to a set
$\AP'$, such that $\AP \subset \AP'$.  We extend $\Traces(K)$ into the
set of traces $\Traces(K,\AP')$ that is agnostic about whether each
new proposition holds at each state.
A trace $(P_0\cup P'_0) (P_1 \cup P'_1) \ldots\in\Traces(K,\AP')$ whenever
$P_0 P_1 \ldots\in\Traces(K)$, and for all $i\geq 0$:
$P'_i \subseteq \AP' \setminus \AP$.
The final conjunct requires every possible set of new atomic
propositions to be included in the traces.

%% file: security.tex
\section{Security Policies in HyperLTL}
\label{sect:ExamplesSecurity}
\label{sec:examples}

We now put HyperLTL into action by formulating several \emph{infor\-mation-flow security policies}, which stipulate how information may propagate from inputs to outputs.  
Information-flow is a very active field in security; see~\cite{SabelfeldM03,FocardiG00} for surveys.

\paragraph*{Noninterference.} 
A program satisfies \emph{noninterference}~\cite{GoguenM82} when the
outputs observed by low-security users are the same as they would be
in the absence of inputs submitted by high-security users.  Since its
original definition, many variants with different execution models
have been named ``noninterference.''  For clarity of our examples, we
choose a simple \emph{state-based} synchronous execution model in
which atomic propositions of the traces contain the values of program
variables, and in which progress of time corresponds to execution
steps in the model.  We also assume
that the variables are partitioned into input and output variables,
and into two security levels, \emph{high} and \emph{low}.  (We could
handle lattices of security levels by conjoining several formulas that
stipulate noninterference between elements of the
lattice.)

\emph{Noninference}~\cite{McLean:1994:GeneralTheory} is a variant of noninterference that 
can be stated in our simple system model.
Noninference stipulates that, for all traces, the low-observable behavior must not change when all high inputs are replaced by a dummy input $\lambda$, that is, when the high input is removed. 
Noninference, a \emph{liveness hyperproperty}~\cite{ClarksonS10}, can be expressed in HyperLTL as follows:
\begin{equation}
\label{hp:noninference}
\forall\pi.\exists\pi'.~(\LTLsquare\lambda_{\pi'})~\wedge~\pi\!=_{L}\!\pi'
\end{equation}
where $\lambda_{\pi'}$ expresses that all of the high inputs in the current state of $\pi'$ are $\lambda$, and $\pi\!=_{L}\!\pi'$ expresses that all low variables in $\pi$ and $\pi'$ have the same values.

\paragraph*{Nondeterminism.} 
Noninterference was introduced for use with deterministic programs.
Nonetheless, nondeterminism naturally arises when program specifications  abstract from implementation details, so many variants of noninterference have been developed for nondeterministic programs. 
We formalize two variants here.

A (nondeterministic) program satisfies \emph{observational determinism}~\cite{ZdancewicM03} if every pair of traces with the same initial low observation remain indistinguishable for low users. 
That is, the program appears to be deterministic to low users. 
Programs that satisfy observational determinism are immune to \emph{refinement attacks}~\cite{ZdancewicM03}, because observational determinism is preserved under refinement. 
Observational determinism, a \emph{safety hyperproperty}~\cite{ClarksonS10}, can be expressed in HyperLTL as follows:
\begin{equation}
\label{hp:od}
\forall \pi.\forall\pi'. ~\pi[0]\!=_{L,\mathsf{in}}\!\pi'[0] ~\rightarrow~ \pi\!=_{L,\mathsf{out}}\!\pi'~
\end{equation}
where $\pi\!=_{L,\mathsf{in}}\!\pi'$ and $\pi\!=_{L,\mathsf{out}}\!\pi'$ express that both traces agree on the low input and low output variables, respectively. 

\emph{Generalized noninterference} (GNI)~\cite{McCullough:1987:GNI} permits nondeterminism in the low-observable behavior, but stipulates that low-security outputs may not be altered by the injection of high-security inputs.
Like noninterference, GNI was original formulated for event-based systems, but it can also be formulated for state-based systems~\cite{McLean:1994:GeneralTheory}.  
GNI is a liveness hyperproperty and can be expressed as follows:
\begin{equation}
\label{hp:gni}
\forall\pi.\forall\pi'.\exists\pi''.~\pi\!=_{H,\mathsf{in}}\!\pi'' ~\wedge~ \pi'\!=_{L}\!\pi''
\end{equation}
%
The trace $\pi''$ in (\ref{hp:gni}) is an \emph{interleaving} of the
high inputs of the first trace and the low inputs and outputs of the
second trace.  Other security policies based on interleavings, such as
\emph{restrictiveness}~\cite{McCullough:1990:Hookup},
\emph{separability}~\cite{McLean:1994:GeneralTheory}, and
\emph{forward correctability}~\cite{Millen94} can similarly be
expressed in HyperLTL.

\paragraph*{Declassification.}
Some programs need to reveal secret information to fulfill functional requirements.  
For example, a password checker must reveal whether the entered password is correct or not.  
The noninterference policies we have examined so far prohibit such behavior. 
More flexible security policies have been designed to permit \emph{declassification} of information; see~\cite{Sabelfeld+Sands/2005/DimensionsAndPrinciplesOfDeclassification} for a survey.

With HyperLTL, we easily specify customized declassification policies. 
For example, suppose that a system inputs a password in its initial state, then declassifies whether that password is correct in the next state.  
The following policy (a safety hyperproperty) stipulates that leaking the correctness of the password is permitted, but that otherwise observational determinism must hold: 
\begin{equation}
\forall \pi.\forall\pi'.  (\pi[0]\!=_{L,\mathsf{in}}\!\pi'[0]~\wedge~\LTLcircle(\mathit{pw}_\pi\!\leftrightarrow\mathit{pw}_{\pi'})) \rightarrow \pi\!=_{L,\mathsf{out}}\!\pi'
\end{equation}
where atomic proposition $\mathit{pw}$ expresses that the entered password is correct. 

\paragraph*{Quantitative noninterference.}
Quantitative information-flow policies~\cite{Gray/1991/TowardAMathematicalFoundationForIFSecurity,ClarkHM05,Koepf+Basin/2007/AnInformationTheoreticModelForAdaptiveSideChannelAttacks,ClarksonMS09} permit leakage of information at restricted rates.  
One way to measure leakage is with \emph{min-entropy}~\cite{Smith/2009/OnTheFoundationsOfQantitativeInformationFlow}, which quantifies the amount of information an attacker can gain given the answer to a single guess about the secret.
The \emph{bounding problem}~\cite{Yasuoka+Terauchi/2010/OnBoundingProblemsOfQuantitativeInformationFlow} for min-entropy is to determine whether that amount is bounded from above by a constant $n$. 
Assume that the program whose leakage is being quantified is deterministic, and assume that the secret input to that program is uniformly distributed.
The bounding problem then reduces to determining that there is no tuple of $2^n+1$ low-distinguishable traces~\cite{Smith/2009/OnTheFoundationsOfQantitativeInformationFlow,Yasuoka+Terauchi/2010/OnBoundingProblemsOfQuantitativeInformationFlow} (a safety hyperproperty).   
We can express that as follows:
\begin{equation}
\neg\exists\pi_0.\;\dots\;.~\exists\pi_{2^n} . ~\Big(\bigwedge_{i}\pi_i=_{L,\mathsf{in}}\pi_0\Big) ~\wedge \bigwedge_{i\not=j}\pi_i\not=_{L,\mathsf{out}}\pi_j
\end{equation}
The initial negation can pushed inside to obtained a proper HyperLTL
formula.

Quantitative flow and entropy naturally bring to mind probabilistic systems.  
We haven't yet explored extending our logics to enable specification of policies that involve probabilities.
Perhaps techniques previously used with epistemic logic~\cite{GrayS98} could be adapted; we leave this as future work.


\paragraph*{Event-based systems.}
Our examples above use a synchronous state-based execution model.
Many formulations of security policies, including the original formulation of noninterference~\cite{GoguenM82}, instead use an \emph{event-based} system model, in which input and output events are not synchronized and have no relation to time. 
HyperLTL can express policies for asynchronous execution models, too. 
For example, HyperLTL can express the original definition of noninterference~\cite{GoguenM82} and observational determinism; 
\seeappendix{Appendix~\ref{sec:eventbased}} shows how.
The key idea is to allow the system to stutter and to quantify over all stuttered versions of the executions. 
We characterize the correct \emph{synchronization} of a pair of traces as having updates to low variables only at the same positions. 
We then add an additional antecedent to the policy formula to require
that only those pairs of traces that are synchronized correctly need
to fulfill the security condition.

%% file: hyperctl.tex
\section{\Hpctl{}}
\label{sec:hyperctl}

HyperLTL was derived from LTL by extending the models of formulas from single traces to sets of traces.
However, like LTL, HyperLTL is restricted to linear time and cannot
express branching-time properties (e.g., all states that succeed the
current state satisfy some proposition).
We show now that a branching-time logic for
hyperproperties could be derived from a branching-time logic for trace
properties, such as
\ctlstr{}~\cite{EmersonH86}.
We call this logic \Hpctl{}.
The key idea is again to use sets instead of singletons as the models
of formulas.

\paragraph*{Syntax.} 
\Hpctl{} generalizes HyperLTL by allowing quantifiers to appear
anywhere within a formula. 
Quantification in {\Hpctl} is over \emph{paths} through a Kripke structure.
A path $p$ is an infinite sequence of pairs of a state and a set of atomic propositions.  
Hence, a path differs from a trace by including a state of the Kripke structure in each element.  
Formally,
$p \in (S\times 2^\AP)^\omega$, where $S$ is the states of the Kripke
structure.
As with traces, $p[i]$ denotes the element $i$ of $p$, and $p[i,\infty]$ denotes the suffix of $p$ beginning with element $i$.
We also define a new notation:  let $p(i)$ be the state in element $i$ of $p$.

In \Hpctl{}, $\pi$ is a \emph{path variable} and $\exists \pi$ is a \emph{path quantifier}.
Formulas of \Hpctl{} are defined by the following grammar:
\[
\begin{array}{llllllllllllllllll}
\varphi &::= & a_{\pi} & \vert & \neg\varphi & \vert & 
\varphi\vee\varphi \vert & \LTLcircle\varphi & \vert & \varphi \U\varphi & \vert & \exists \pi.\; \varphi & 
\end{array}
\]
We introduce all the syntactic sugar for derived logical operators,
as for HyperLTL.
The universal quantifier can now be defined as syntactic sugar, too: $\forall \pi.\; \varphi \equiv \neg\exists\pi.\; \neg\varphi$. 
A \Hpctl{} formula is \emph{closed} if all occurrences of some path variable
$\pi$ are in the scope of a path quantifier. A \Hpctl{}
\emph{specification} is a Boolean combination of closed \Hpctl{}
formulas each beginning with a quantifier (or its negation).


\paragraph*{Semantics.} 

The validity judgment for \Hpctl{} formulas is written $\Pi \models_K \varphi$, where $K$ is a Kripke structure, and $\Pi:\mathcal{V}\to (S\times 2^{\AP})^{\omega}$ is a \emph{path assignment}, which is a partial function mapping path variables to paths. 
We write $K$ as a subscript on $\models$, because $K$ propagates unchanged through the semantics; we omit $K$ when it is clear from context.
Validity is defined as follows:
\[
\begin{array}{l@{\hspace{1em}}c@{\hspace{1em}}l}
  \Pi\models_K a_{\pi} & \text{iff} & a\in L\big(\Pi(\pi)(0)\big) \\
  \Pi\models_K \neg \varphi & \text{iff} & \Pi\not\models_K\varphi \\
  \Pi\models_K \varphi_1 \vee \varphi_2 & \text{iff} & \Pi\models_K\varphi_1 \text{ or } \Pi\models\varphi_2 \\
  \Pi\models_K\LTLcircle\varphi & \text{iff} & \Pi[1,\infty]\models_K\varphi \\
  \Pi\models_K\varphi_1\U\varphi_2 & \text{iff} & \text{there exists $i \geq 0 :$} \Pi[i,\infty]\models_K\varphi_2 \\
  && \quad\text{and for all $0 \leq j < i$ we have $\Pi[j,\infty]\models_K\varphi_1$}  \\
  \Pi\models_K\exists\pi.\; \varphi & \text{iff} & \text{there exists } p \in \Paths(K,\Pi(\pi')(0)) : \Pi[\pi \mapsto p] \models_K \varphi
\end{array}
\]
In the clause for existential quantification, $\pi'$ denotes the path variable most recently added to $\Pi$ (i.e., closest in scope to $\pi$).
If $\Pi$ is empty, let $\Pi(\pi')(0)$ be the initial state of $K$. 
It would be straightforward but tedious to further formalize this notation, so we omit the details. 
That clause uses another new notation, $\Paths(K,s)$, which is the set of paths produced by Kripke structure $K$ beginning from state $s$.
Formally, $\Paths(K,s)$ contains path $p$, where $p = (s_0, P_0) (s_1, P_1) \ldots$ and $P_i \in 2^\AP$, iff there exists a sequence $s_0 s_1 \ldots$ of states, such that $s_0$ is $s$, and for all $i \geq 0$, it holds that $s_{i+1}\in\delta(s_{i})$ and $P_i = L(s_i)$.

Like with $\Traces$ in HyperLTL, we define $\Paths(K,s,\AP')$ as follows:
We have $(s_0,P_0\cup P'_0) (s_1,P_1 \cup P'_1) \ldots\in\Paths(K,s,\AP')$ iff $(s_0,P_0) (s_1,P_1) \ldots\in\Paths(K,s)$, and for all $i\geq 0$, it holds that $P'_i \subseteq \AP' \setminus \AP$.  

We say that a Kripke structure
$K$ \emph{satisfies} a \Hpctl{} specification $\varphi$, denoted by $K\models
\varphi$, if $\Pi \models_K \varphi$ holds true for the empty assignment.  The
\emph{model checking problem} for \Hpctl{} is to decide
whether a given Kripke structure satisfies a given \Hpctl{} specification.

\paragraph*{\Hpctl{} vs. HyperLTL.} 
LTL can be characterized as the fragment of \ctlstr{} containing formulas  of the form $\A \varphi$, where $\A$ is the \ctlstr{} universal path quantifier and $\varphi$ contains no quantifiers. 
Formula A$\varphi$ is satisfied in \ctlstr{} by a Kripke structure iff $\varphi$ is satisfied in LTL by the traces of the Kripke structure. 

A similar relationship holds between HyperLTL and \Hpctl{}:
HyperLTL can be characterized as the fragment of \Hpctl{} containing formulas in \emph{prenex} form---that is, a series of quantifiers followed by a quantifier-free formula.
A formula $\varphi$ in prenex form is satisfied in \Hpctl{} by a Kripke structure iff $\varphi$ is satisfied in HyperLTL by the traces of the Kripke structure.
\Hpctl{} is a strict generalization of HyperLTL, which extends HyperLTL with the capability to use quantified formulas as subformulas in the scope of temporal operators. 
For example, consider the program $(l := 0 ~\talloblong~ l := 0) ~\talloblong~ (l := 1 ~\talloblong~ l := 1)$.
A low-observer can infer which branch of the center-most nondeterministic choice is taken, but not which branch is taken next.
This is expressed by \Hpctl{} formula $\forall \pi.\ \LTLcircle\, \forall \pi'.\ \LTLcircle\, (l_{\pi}\leftrightarrow l_{\pi'})$. 
There is no equivalent HyperLTL formula. 

As we show in Subsection~\ref{ssect:SecLTL}, the temporal logic SecLTL~\cite{Dimitrova:2012:SecLTL} can be encoded in \Hpctl{}, but not in HyperLTL. 
This provides further examples that distinguish HyperLTL and \Hpctl{}.

%% file: relatedlogics.tex

\section{Related Logics}
\label{sect:relations}
\label{sec:relatedlogics}

We now examine the expressiveness of HyperLTL
and \Hpctl{} compared to several existing temporal logics: LTL, {\ctlstr}, QPTL, ETL, and SecLTL.
There are many other logics that we could compare to in future work; some of those
are discussed in Section~\ref{sec:relatedwork}.


\subsection{Temporal Logics}
\label{sec:ltlctl}

\Hpctl{} is an extension of \ctlstr{} and therefore
subsumes LTL, CTL, and \ctlstr{}.  Likewise, HyperLTL subsumes
LTL.  But temporal logics LTL, CTL, and \ctlstr{} cannot express information-flow policies.
LTL formulas express properties of individual execution paths.
All of the noninterference properties of Section~\ref{sec:examples}
are properties of sets of execution
paths~\cite{McLean:1994:GeneralTheory,ClarksonS10}.  Explicit path quantification does
enable their formulation in HyperLTL. 

Even though
CTL and \ctlstr{} have explicit path quantifiers, information-flow security policies, such as observational determinism~\eqref{hp:od},
cannot be expressed with them.
Consider the following fragment of \ctlstr{} semantics:
\begin{equation*}
\begin{array}{ll}
s \models \A \varphi &~\text{iff for all $p \in \Paths(K,s):~p \models \varphi$} \\
p \models \Phi &~\text{iff $p(0) \models \Phi$}
\end{array}
\end{equation*}
\emph{Path} formulas $\varphi$ are modeled by paths $p$, and \emph{state} formulas $\Phi$ are modeled by states $s$.
State formula $\A \varphi$ holds at state $s$ when all paths proceeding from $s$ satisfy $\varphi$.
Any state formula $\Phi$ can be treated as a path formula, in which case $\Phi$ holds of the path iff $\Phi$ holds in the first state on that path.
Using this semantics, consider the meaning of $\AAx \varphi$, which is the form of observational determinism~\eqref{hp:od}:
\begin{equation*}
\begin{array}{l}
s \models \AAx \varphi \\
= \text{for all $p \in \Paths(K,s):~p \models A\varphi$} \\
= \text{for all $p \in \Paths(K,s)$ and $p' \in \Paths(K,s):~p' \models \varphi$} 
\end{array}
\end{equation*}
Note how the meaning of $\AAx \varphi$ is ultimately determined by the meaning of $\varphi$, where $\varphi$ is modeled by the single path $p'$.  
Path $p$ is ignored in determining the meaning of $\varphi$;
the second universal path quantifier causes $p$ to ``leave scope.''
Hence $\varphi$ cannot express correlations between $p$ and $p'$, as observational determinism requires.
So \ctlstr{} path quantifiers do not suffice to express information-flow policies.
Neither do CTL path quantifiers, because CTL is a sub-logic of \ctlstr{}. 
In fact, even the modal $\mu$-calculus does not suffice to express some information-flow properties~\cite{AlurCZ06}.

By using the self-composition construction~\cite{BartheDR04,Terauchi+Aiken/05/SecureInformationFlowAsSafetyProblem}, it is possible to express
relational noninterference in
CTL~\cite{BartheDR04}
and observational determinism in
\ctlstr{}~\cite{HuismanWS/06/TLCharacterisationOfOD}.  Those approaches
resemble \Hpctl{}, but \Hpctl{} formulas express policies
directly over the original system, rather than over a self-composed
system.  Furthermore, the self-composition approach does not seem
capable of expressing policies that require both universal and
existential quantifiers over infinite executions,
like noninference~\eqref{hp:noninference} and generalized noninterference~\eqref{hp:gni}.
It is straightforward to express such policies in our logics.

\paragraph*{QPTL.}
Quantified propositional temporal logic
(QPTL)~\cite{Sistla+Vardi+Wolper/1987/TheComplementationProblemForBuchiAutomata}
extends LTL with quantification over propositions, whereas
HyperLTL extends LTL with quantification over traces.  
Quantification over traces is more powerful than quantification over
propositions, as we now show.

QPTL formulas are generated by the following grammar,
where $a\in\AP$:
\[
\psi ~~ ::= ~~ a ~~\vert~~ \neg\psi ~~\vert~~
\psi\vee\psi
~~\vert~~ \LTLcircle\psi ~~\vert~~ \LTLdiamond\psi
~~\vert~~ \exists a.\; \psi
\] 
All QPTL connectives have the same semantics as in LTL, except 
for propositional quantification:
\[
p\models\exists a.\psi\;\;\;\text{ iff }\;\;\; \text{there exists }  p'\in (2^{\AP})^{\omega}:~p=_{\AP\setminus a}p' \text{ and }p'\models\psi\;.
\]

\begin{theorem}
HyperLTL subsumes QPTL, but QPTL does not subsume HyperLTL.
\end{theorem}

\begin{proof}[\textit{Proof sketch}]
To express a QPTL formula in HyperLTL, rewrite the formula to prenex form, and rename all bound propositions with unique fresh names from a set $\AP'$.
These propositions act as free variables, which are unconstrained because they do not occur in the Kripke structure.
Replace each propositional quantification $\exists a$ in the QPTL formula by a path quantification $\exists \pi_a$ in the
HyperLTL formula.
And replace each occurrence of $a$ by $a_{\pi_a}$.
The result is a HyperLTL formula that holds iff the original QPTL formula holds.

But not all HyperLTL formulas can be expressed in QPTL.  
For example, QPTL cannot express properties that require the
existence of paths, such as $\exists \pi . \LTLcircle a_\pi$.
\end{proof}

In Section~\ref{sect:MCandSat}, we exploit the relationship
between HyperLTL and QPTL to obtain a model checking algorithm for
HyperLTL.

%

\subsection{Epistemic Logics}

HyperLTL and \Hpctl{} express information-flow policies by explicit quantification over multiple traces or paths. 
Epistemic temporal logic has also been used to express such policies~\cite{HalpernO08,Balliu::Epis,Chadha+al/2009/EpistemicLogicsAppliedPiCalculus,vanderMeyden+Wilke/2007/PreservationOfEpistemicPropertiesInSecurityProtocolImplementations} by implicit quantification over traces or paths with the \emph{knowledge} connective ${\sf K}$ of epistemic logic~\cite{Fagin-book-95}.
We do not yet know which is more powerful, particularly for information-flow policies.  
But we do know that HyperLTL subsumes a common epistemic temporal logic.

Define ETL (epistemic temporal logic) to be LTL with the addition of ${\sf K}$ under its perfect recall semantics~\cite{Meyden/1993/AxiomsForKnowledgeAndTimeInDistributedSystemsWithPerfectRecall,Balliu::Epis,Fagin-book-95}.
The model of an ETL formula is a pair $(K,\Agts)$ of a Kripke structure $K$ and a set $\Agts$ of equivalence relations on $\AP$, called the \emph{agents}; each relation models the knowledge of an agent.
(\emph{Interpreted systems}, rather than Kripke structures, are often used to model ETL formulas~\cite{Meyden/1993/AxiomsForKnowledgeAndTimeInDistributedSystemsWithPerfectRecall,Fagin-book-95}.
Interpreted systems differ in style but can be translated to our formulation.)
In the \emph{asynchronous} semantics of ETL, ${\sf K}_{A}\psi$ holds on state $i$ of trace $t\in\Traces(K)$, denoted $t,i\models {\sf K}_A\varphi$, iff
\[
\begin{array}{r}
\text{for all }t'\in\Traces(K) :~ t[0,i]\!\approx_A\! t'[0,i] \text{ implies } t',i\!\models\varphi,
\end{array}
\]
where $\approx_A$ denotes stutter-equivalence on finite traces with respect to $A$. 
%
In the \emph{synchronous} semantics of ETL, stutter-equivalence is replaced by stepwise-equivalence.

The following two theorems show that HyperLTL subsumes ETL:


\begin{theorem}\label{thm:synchEpistemic}
In the synchronous semantics, for every ETL formula $\psi$ and every set $\Agts$ of agents, there exists a HyperLTL formula $\varphi$ such that for all Kripke structures $K$, we have $(K,\Agts)\models\psi$ iff $K\models\varphi$. 
\end{theorem}

\newcounter{thm-asynchEpistemic}
\setcounter{thm-asynchEpistemic}{\value{theorem}}
\begin{theorem}\label{thm:asynchEpistemic}
In the asynchronous semantics, for every ETL formula $\psi$ and every set $\Agts$ of agents, there exists a HyperLTL formula $\varphi$ such that for all asynchronous Kripke structures $K$, we have $(K,\Agts)\models\psi$ iff $K\models\varphi$.  
\end{theorem}

\noindent Proofs of both theorems appear in \seeappendix{Appendix~\ref{app:asynchEpistemic}}. Theorem~\ref{thm:asynchEpistemic} requires an additional assumption that $K$ is an \emph{asynchronous} Kripke structure, i.e. that it can always stutter in its current state and that it is indicated in an atomic proposition whether the last state was a stuttering step. 

 
HyperLTL and ETL have the same worst-case complexity for model checking, which is non-elementary. 
But, as we show in Section~\ref{sec:ctlmodelchecking}, the complexity of our model checking algorithm on the information-flow policies of Section~\ref{sec:examples} is much better---only NLOGSPACE (for observational determinism, declassification, and quantitative noninterference for a fixed number of bits) or PSPACE (for noninference and generalized noninterference) in the size of the system.  
For those policies in NLOGSPACE, that complexity, unsurprisingly, is as good as algorithms based on self-composition~\cite{BartheDR04}.
This ability to use a general-purpose, efficient HyperLTL model checking algorithm for information flow seems to be an improvement over encodings of information flow in ETL. 

\subsection{SecLTL}
\label{ssect:SecLTL}

SecLTL~\cite{Dimitrova:2012:SecLTL} extends
temporal logic with the \emph{hide} modality $\Hide$, which allows to express
information flow properties such as
noninterference~\cite{GoguenM82}.
%
%
The semantics of SecLTL
is
defined in terms of labeled transition system, where the edges are
labeled with valuations of the set of variables.  
The formula $\Hide_{H,O}\varphi$
specifies that the current valuations of a subset $H$ of the
input variables $I$ are kept secret from an attacker who may observe
variables in $O$ until the release condition $\varphi$ becomes true.
The semantics is formalized in terms of a set of \emph{alternative paths}
to which the \emph{main path} is compared: 
\[\begin{array}{l}
\AltPaths(p,H) = \{p'\in\Paths(K_M,p[0]) \mid p[1]\!=_{I\setminus H}\!p'[1] \text{ and } p[2,\infty]\!=_{I}\!p'[2,\infty]\}
\end{array}
\]
where $K_M$ is the equivalent Kripke structure for the labeled transition system $M$ (we will explain the translation later in this section.) 
A path $p$ satisfies the SecLTL formula $\Hide_{H,O}\varphi$, denoted by $p\models \Hide_{H,O}\varphi$, iff
\[
\begin{array}{rl}
\forall p'\in\AltPaths(p,H). &\big(p\!=_{O}\!p' ,\text{ or there exists } i\geq0 : \\
	&\qquad p[i,\infty]\models_K\varphi ~ \text{ and }~ p[1,i\!-\!1]\!=_{O}\!p'[1,i\!-\!1]\big)
\end{array}
\]
A labeled transition system $M$ satisfies a SecLTL formula $\psi$, denoted by $M \models \psi$, if every path $p$ starting in the initial state satisfies $\psi$.

SecLTL can express properties like the dynamic creation of secrets
discussed in Section~\ref{sec:hyperctl}, which cannot be expressed by
HyperLTL. However, SecLTL is subsumed by \Hpctl{}.  To encode the
hide modality in \Hpctl{}, we first translate $M$ into a Kripke
structure $K_M$, whose states are labeled with the valuation of the
variables on the edge leading into the state.
The initial state is labeled with the empty set.  In the modified
system, $L(p[1])$ corresponds to the current labels.  
%
We encode $\Hide_{H,O}\varphi$ as the following \Hpctl{} formula: 
\[
\forall\pi'.\;\pi[1]\!=_{I\setminus H}\!\pi'[1] \wedge \LTLcircle\big(\pi[1]\!=_{O}\!\pi'[1]\;\W\;(\pi[1]\!\not=_{I}\!\pi'[1]\vee \varphi)\big) 
\]


\begin{theorem}
For every SecLTL formula $\psi$ and transition system $M$, there  is a \Hpctl{} formula $\varphi$ such that 
$M\models\psi$ iff $K_M\models\varphi$. 
\end{theorem}

The model checking problem for SecLTL is PSPACE-hard in the size of
the Kripke
structure~\cite{Dimitrova:2012:SecLTL}.
The encoding of Sec\-LTL specifications in \Hpctl{} implies that the
model checking problem for \Hpctl{} is also PSPACE-hard (for
a fixed specification of alternation depth $\geq 1$), as claimed in
Theorem~\ref{thm:PSPACElowerBound}.




%% file: modelcheckingsat.tex

\section{Model Checking and Satisfiability}
\label{sect:MCandSat}
\label{sec:ctlmodelchecking}


In this section we exploit the connection between \Hpctl{} and
QPTL to obtain a model checking algorithm for \Hpctl{} and study
its complexity.  
We identify a hierarchy of fragments of \Hpctl{} characterized by the number of quantifier alternations. 
This hierarchy allows us to give a precise characterization of the complexity of the model checking problem. 
The fragment of formulas with quantifier alternation
depth 0 includes already many formulas of interest and our result provides an
NLOGSPACE algorithm in the size of the Kripke structure.


\vspace{4pt}\begin{definition}[Alternation Depth] A \Hpctl{} 
  formula $\varphi$ in NNF has alternation depth 0 plus the highest
  number of alternations from existential to universal and universal
  to existential quantifiers along any of the paths of the formula's
  syntax tree starting in the root.  Occurrences of $\U$ and $\R$ 
	count as an additional alternation.
\end{definition}

\begin{theorem}
\label{thm:nonelementaryformula}
The model checking problem for \Hpctl{} specifications $\varphi$ with
alternation depth $k$ on a Kripke structure $K$ is complete for NSPACE$(g_c(k,|\varphi|))$ and it is in NSPACE$(g_c(k-1,|K|))$ for some $c>0$.
\end{theorem}


The function $g_c(x,y)$ denotes a tower of exponentials of height $x$ with argument $y$: $g_c(0,y)=y$ and $g_c(x,y)=c^{g(x-1,y)}$. NSPACE$(g_c(x,y))$ denotes the class of languages accepted by a Turing machine bounded in space by $O(g_c(x,y))$. Abusing notation, we define $g_c(-1,y)=\log y$ and NSPACE$(\log y)=\text{NLOGSPACE}$ in $y$.

\begin{proof}
  Both directions, the lower bound and the upper bound, are based on
  the complexity of the satisfiability problem for QPTL formulas
  $\varphi$ in prenex normal form and with alternation depth $k$,
  which is complete for
  NSPACE$(g(k,|\varphi|))$~\cite{Sistla+Vardi+Wolper/1987/TheComplementationProblemForBuchiAutomata}.

  For the upper bound on the \Hpctl{} model checking complexity,
  we first translate until operators $\psi \U
  \psi'$~as 
  $ \exists t. ~t \wedge \LTLsquare(t
  \to\psi'\vee(\psi\wedge\LTLcircle t)) \wedge \neg\LTLsquare t$.
	Let $\psi(K,\AP')$ encode a Kripke structure $K$, where $K=(S,s_0,\delta,\AP,L)$, as a QPTL formula 
	(cf.~\cite{Manna+Pnueli/1992/TheTemporalLogicOfReactiveSystems}) using the set of atomic propositions $\AP'$, which must contain atomic propositions replacing those of $\AP$ and additional atomic propositions to describe the  states $S$. 
	The formula $\psi(K,\AP')$ is linear in $|K|$ and does not require additional quantifiers. 
	
  \Hpctl{} path quantifiers $\exists\pi.\varphi$ and $\forall\pi.\varphi$ are then encoded as $ \exists \AP_\pi.$ $\psi(K,\AP_\pi)\wedge\varphi_{\AP_\pi}$ and $\forall\AP_\pi.\psi(K,\AP_\pi)\to\varphi_{\AP_\pi}$, where $\AP_\pi$ is a set of \emph{fresh} atomic propositions	including a copy of $\AP$ and additional atomic propositions to describe the states $S$. 
	The formula $\varphi_{\AP_\pi}$ is obtained from $\varphi$ by 
	replacing all atomic propositions referring to path $\pi$ by their copies in $\AP_\pi$. 
	Atomic propositions in the formula that are not in $\AP$ (i.e. their interpretation is not fixed in $K$) need to be added to the sets $\AP_{\pi}$ accordingly.


For the lower bound, we reduce the satisfiability problem for a given
QPTL formula $\varphi$ in prenex normal form to a model checking
problem $K \models \varphi'$ of
\Hpctl{}. We assume, without loss of generality,
that $\varphi$ is closed (if a free proposition occurs in $\varphi$,
we bind it with an existential quantifier) and each quantifier in
$\varphi$ introduces a different proposition.

The Kripke structure $K$ consists of two states $S=\{s_0, s_1\}$, is fully connected $\delta(s)=S$ for all $s \in S$, and has a single atomic proposition $\AP=\{p\}$. The states are labeled as follows: $L(s_0)=\emptyset$ and $L(s_1)= \{p\}$. 
Essentially, paths in $K$ can encode all sequences of
valuations of a variable in QPTL.  
To obtain the \Hpctl{} formula, we now simply replace every quantifier in the QPTL formula with a path quantifier. 
The only technical problem left is that quantification in QPTL allows to choose freely the value of $p$ in the current state, while path quantification in \Hpctl{} only allows
the path to differ in the next state. 
We solve the issue by shifting the propositions using a next operator. 
%
\end{proof}

\paragraph{Lower bounds in $|K|$.}
An NLOGSPACE \emph{lower} bound in the size of the Kripke structure for fixed
specifications with alternation depth 0 follows from the non-emptiness
problem of non-deterministic B\"uchi automata.  For alternation depth
1 and more we can derive PSPACE hardness in the size of the Kripke
structure from the encoding of the logic SecLTL into \Hpctl{} (see
Subsection~\ref{ssect:SecLTL}).

The result can easily be transferred to HyperLTL, since in the SecLTL formula that is used to prove PSPACE hardness, the Hide operator does not occur in the scope of temporal operators and hence the translation yields a HyperLTL formula. 

\begin{theorem}
\label{thm:PSPACElowerBound}
For HyperLTL formulas the model checking problem is hard for PSPACE in
the size of the system.
\end{theorem}



\paragraph*{A Remark on Efficiency}
The use of the standard encoding of the until operator in QPTL with an
additional quantifier shown above is, in certain cases, wasteful. The
satisfiability of QPTL formulas can be checked with an
automata-theoretic construction, where we first transform the formula
into prenex normal form, then generate a nondeterministic B\"uchi
automaton for the quantifier-free part of the formula, and finally
apply projection and complementation to handle the existential and
universal quantifiers.  In this way, each quantifier alternation,
including the alternation introduced by the encoding of the until
operators, causes an exponential blow-up.  However, if an until
operator occurs in the quantifier-free part, the standard
transformation of LTL formulas to nondeterministic B\"uchi automata
handle this until operator without requiring a quantifier elimination,
resulting in an exponential speedup.

Using this insight, the model checking complexity for many of the formulas presented above and in Section~\ref{sect:ExamplesSecurity} can be reduced by one exponent. 
Additionally, the complexity with respect to the size of the system reduces to NLOGSPACE
for \Hpctl{} formulas where the leading quantifiers are all of the same type
and are followed by some quantifier-free formula which may contain until
operators without restriction. Observational determinism and the declassification policy discussed in Section~\ref{sect:ExamplesSecurity} are examples for specifications in this fragment.
This insight was used for the prototype implementation described in Section~\ref{sec:prototype} and it avoids an additional complementation step for noninference~\eqref{hp:noninference}. 

%

\paragraph*{Satisfiability.}
The positive result regarding the model checking problem for \Hpctl{} does not carry over to the satisfiability
problem.  
The \emph{finite-state satisfiability problem} consists of
the existence of a finite model, while the \emph{general
  satisfiability problem} asks for the existence of a possibly infinite model. 
\begin{theorem}
  For \Hpctl{}, finite-state satisfiability is hard for $\Sigma^0_1$ and
  general satisfiability is hard for $\Sigma^1_1$.
\end{theorem}
\noindent In the proof, located in \seeappendix{Appendix~\ref{app:satisfiability}}, we reduce the LTL synthesis problem of distributed systems to the satisfiability problem of \Hpctl{}.

%% file: prototype.tex
\section{Prototype Model Checker}
\label{sec:prototype}

The results of the previous section yield a model checking algorithm for all of \Hpctl{}.  
But most of our information-flow policy examples do not require the full expressiveness of \Hpctl{}.  
In fact, we have been able implement a prototype model checker for an expressive fragment of the logic mostly using off-the-shelf components.

Define \emph{\Hp} as the fragment of HyperLTL (and of \Hpctl{}) in which the series of quantifiers at the beginning of a formula may involve at most one alternation.
Every formula in {\Hp} thus may begin with at most two (whence the name) kinds of quantifiers---a sequence of $\forall$'s followed by a sequence of $\exists$'s, or vice-versa.
For example, $\exists\pi.\psi$ and $\forall\pi_1.\forall\pi_2.\exists\pi_3.\psi$ are allowed, but $\forall\pi_1.\exists\pi_2.\forall\pi_3.\psi$ is not.
{\Hp} suffices to express all the security policies formulated in Section~\ref{sect:ExamplesSecurity}. 
(Another logic for hyperproperties, $\mathcal{IL}^k_\mu$~\cite{MilushevC13}, similarly restricts fixpoint operator alternations with no apparent loss in expressivity for security policies.)

Our model checking algorithm for {\Hp}, detailed in \seeappendix{Appendices~\ref{sec:constructions} and \ref{sec:prototypedetails}}, is based on algorithms for LTL model checking~\cite{GerthVardi:1995:OntheFlyLTL,Vardi:1996:LTL,Gastin:2001:FastLTLtoBuchi}.
Those LTL algorithms determine whether a Kripke structure satisfies an LTL formula by performing various automata constructions and by checking language containment.
Our algorithm likewise uses automata constructions and language containment, as well as self composition~\cite{BartheDR04,Terauchi+Aiken/05/SecureInformationFlowAsSafetyProblem} and a new \emph{projection} construction.

We prototyped this algorithm in about 3,000 lines of OCaml code. 
Our prototype accepts as input a Kripke structure and a {\Hp} formula,  then constructs the automata required by our algorithm, and  outputs a countermodel if the formula does not hold of the structure. 
For automata complementation, our prototype outsources to GOAL~\cite{GOAL:2007}, an interactive tool for manipulating {\Buchi} automata. 
We have used the prototype to verify noninference~\eqref{hp:noninference}, observational determinism~\eqref{hp:od}, and generalized noninterference~\eqref{hp:gni} for small Kripke structures (up to 10 states); running times were about 10 seconds or less.

Since our algorithm uses automata complementation, the worst-case running time is exponential in the size of the Kripke structure's state space and doubly exponential in the formula size.
So as one might expect, our prototype currently does not scale to medium-sized Kripke structures (up to 1,000 states).  
But our purpose in building this prototype was to demonstrate a proof-of-concept for model checking of hyperproperties.
We conjecture that practical symbolic model checking algorithms, such as BMC and IC3, could be used to scale up our approach to real-world systems.

%% file: relatedwork.tex
\section{Related Work}
\label{sec:relatedwork}

McLean~\cite{McLean:1994:GeneralTheory} formalizes security policies as  closure with respect to \emph{selective interleaving functions}. 
He shows that trace properties cannot express security policies such as noninterference and average response time, because those are not properties of single execution traces.
Mantel~\cite{Mantel/2000/PossibilisticDefinitionsofSecurityAnAssemblyKit} formalizes security policies with \emph{basic security predicates}, which stipulate \emph{closure conditions} for trace sets.

Clarkson and Schneider~\cite{ClarksonS10} introduce \emph{hyperproperties}, a framework for expressing security policies. 
Hyperproperties are sets of trace sets, and are able to formalize security properties such as noninterference, generalized noninterference, observational determinism and average response time. 
Clarkson and Schneider use second-order logic to formulate hyperproperties.  
That logic isn't verifiable, in general, because it cannot be effectively and completely axiomatized. 
Fragments of it, such as HyperLTL and \Hpctl{}, can be verified.

Alur et al.~\cite{AlurCZ06} show that modal $\mu$-calculus is insufficient to express all \emph{opacity} policies~\cite{BryansKMR05}, which prohibit observers from discerning the truth of a predicate.
(Alur et al.~\cite{AlurCZ06} actually write ``secrecy'' rather than ``opacity.'')
Simplifying definitions slightly, a trace property $P$ is \emph{opaque} iff for all paths $p$ of a system, there exists another path $p'$ of that system, such that $p$ and $p'$ are low-equivalent, and exactly one of $p$ and $p'$ satisfies $P$.
Noninference~\eqref{hp:noninference} is an opacity policy~\cite{PeacockR06} that HyperLTL can express.



Huisman et al. reduce observational-determinism properties to properties in $\text{CTL}^*$~\cite{HuismanWS/06/TLCharacterisationOfOD} and in modal $\mu$-cal\-cu\-lus~\cite{HuismanB12} on a self-composed system.
Barthe et al. use self composition to verify observational determinism~\cite{BartheDR04} and noninterference~\cite{BartheCK13} on terminating programs.
Van der Meyden and Zhang~\cite{VanDerMeyden:2007:verifOfNonInterf} reduce a broader class of information-flow policies to safety properties on a self-composed system expressible in standard linear and branching time logics, and use model checking to verify noninterference policies.
Their methodology requires customized model checking algorithms for each security policy, whereas this work proposes a single algorithm for all policies.



Balliu et al.~\cite{Balliu::Epis} use a linear-time temporal epistemic logic to specify many declassification policies derived from noninterference.
Their definition of noninterference, however, seems to be that of observational determinism~\eqref{hp:od}.  
They do not consider any information-flow policies involving existential quantification, such as noninference. 
They also do not consider systems that accept inputs after execution has begun.
Halpern and O'Neill~\cite{HalpernO08} use a similar temporal epistemic logic to specify \emph{secrecy} policies, which subsume many definitions of noninterference; they do not pursue model checking algorithms. 

Alur et al.~\cite{Alur+Cerny+Chaudhuri/07/MCTreesWithPathEquivalences} discuss 
branching-time logics with \emph{path equivalences}
that are also able to express certain security properties.  The
authors introduce operators that resemble the knowledge operator of
epistemic logics.  As the logics build on branching-time
logics they are not subsumed by HyperLTL.  The relationship to
\Hpctl{} is still open.

Milushev and Clarke~\cite{MilushevC12,Milushev:2013:thesis,MilushevC13} propose three logics for hyperproperties:
\begin{compactitem}
\item \emph{Holistic hyperproperty logic} $\mathcal{HL}$, which is based on coinductive predicates over streams.  Holistic hyperproperties ``talk about whole traces at once; their specifications tend to be straightforward, but they are difficult to reason about, exemplified by the fact that no general approach to verifying such hyperproperties exists''~\cite{MilushevC12}. HyperLTL and {\Hpctl} are logics that talk about whole traces at once, too; and they have straightforward specifications as well as a general approach to verification.  
\item \emph{Incremental hyperproperty logic} $\mathcal{IL}$ is a fragment of \emph{least fixed-point logic}~\cite{BradfieldS07}.  There is a manual verification methodology for $\mathcal{IL}$~\cite{MilushevC12}, but no automated decision procedure.
\item Another incremental hyperproperty logic $\mathcal{IL}^k_\mu$, a fragment of polyadic modal $\mu$-calculus~\cite{Andersen94} that permits at most one quantifier alternation (a greatest fixed-point followed by a least fixed-point).  There is an automated model checking technique~\cite{MilushevC13} for $\mathcal{IL}^k_\mu$ based on \emph{parity games}.  That technique has been prototyped and applied to a few programs.
\end{compactitem}

\noindent All these logics suffice to express security policies such as noninterference 
and generalized noninterference. Like our logics, the exact expressive limitation is 
still an open problem. 

As the preceding discussion makes clear, the expressiveness of HyperLTL and
{\Hpctl} versus several other logics is an open question. 
It's possible that some of those logics will turn out to be more
expressive or more efficiently verifiable than HyperLTL or {\Hpctl}.
It's also possible that it will turn out to be simply a matter of taste which style of logic is
more suitable for hyperproperties.  
The purpose of this paper was to explore one design option:
a familiar syntax, based on widely-used temporal logics, that can straightforwardly express
well-known hyperproperties.

%% file: conclusion.tex

\section{Concluding Remarks}
\label{sec:conclusion}

In designing a logic for hyperproperties, starting with HyperLTL was natural, because hyperproperties are sets of trace sets, and LTL uses trace sets to model programs.  
From HyperLTL, the extension to \Hpctl{} was also natural: we simply removed the restrictions on where quantifiers could appear.  
%
The curtailment to {\Hp} was also natural, because it was the fragment needed to express information-flow security policies.  
{\Hp} permits up to one quantifier alternation, but what about hyperproperties with more? 
We do not yet know of any security policies that are examples.
As Rogers~\cite{Rogers87} writes, ``The human mind seems limited in its ability to understand and visualize beyond four or five alternations of quantifier.  
Indeed, it can be argued that the inventions{\ldots}of mathematics are devices for assisting the mind in dealing with one or two additional alternations of quantifier.''
For practical purposes, we might not need to go much higher than one quantifier alternation.


%% file: extendedEncodings.tex

\section{Event-based Execution Model}
\label{sec:eventbased}

\subsection{Goguen and Meseguer's noninterference}

\paragraph*{Noninterference.} 
A point of reference for most of the literature on information flow is the definition of \emph{noninterference} that was introduced by Goguen and Meseguer in 1982 \cite{GoguenM82}. 
In this subsection, we show how to express noninterference in a simple HyperLTL formula. 

The system model used in \cite{GoguenM82}, which we will refer to as \emph{deterministic state machines}, operates on commands $c\in C$ that are issued by different users $u\in U$. 
The evolution of a deterministic state machine is governed by the transition function $\DO:S\times U\times C\to S$ and there is a separate observation function $\out:S\times U\to \Out$ that for each user indicates what he can observe. 


We define standard notions on sequences of users and events. 
For $w\in(U\times C)^*$ and $G\subseteq U$ let $|w|_{G}$ denote the projection of $w$ to the commands issued by the users in $G$. 
Further, we extend the transition function $\DO$ to sequences, $\DO(s,(u,c).w)=\DO(\DO(s,u,c),w)$, where the dot indicates concatenation. 
Finally, we extend the observation function $\out$ to sequences $w$, indicating the observation \emph{after} $w$: $\out(w,G)=\out(\DO(s_0,w),G)$.
\emph{Noninterference} is then defined as a property on systems $M$. A set of users $G_H\subseteq U$ does not interfere with a second group of users $G_L\subseteq U$, if
\[
\forall w\in (U\times C)^*. ~\out(w,G_L)=\out(|w|_{G_H},G_L)
\]

That is, we ask whether the same output would be produced by the system, if all actions issued by any user in $G_H$ were removed.

\paragraph*{Encoding GM's System Model}

First, we need to map their system model into Kripke structures as used for the formulation of \Hpctl{} (which, obviously, must not solve problem by itself). 
We choose an intuitive encoding of state machines in Kripke structures that indicates in every state (via atomic propositions) which observations can be made for the different users, and also which action was issued last, including the responsible user.

We choose a simple translation that maps a state machine  $M=(S,U,C,\Out,$ $\out,\DO,s_0)$ to the Kripke structure $K=(S',s'_0,\delta,\AP,L)$, where $S'=\{s_0\}\cup S\times U\times C$, $s'_0=s_0$, $\AP= U \times C ~\cup~ U\times\Out$, and the labeling function is defined as $L(s_0)=\{(u,\out(s_0,u))\mid u\in U\}$ for the initial state and $L((s,u,c))=\{(u,c)\}\cup\{(u',\out(s,u'))\mid u'\in U\}$ for all other states. 

The transition function is defined as 
\[
\delta(s,u,c) = \{(s',u',c') \mid \DO(s,u',c')=s'\!,~ u'\!\in U,~ c'\!\in C \}
\]

Each state (except for the initial state) has labels indicating the command that was issued last, and the user that issued the command. 
The remaining labels denote the observations that the individual users can make in this state. 

To access these two separate pieces of information, we introduce  functions $\Input:S\to U\times C$, which is not defined for $s_0$, and $\out:S\times U\to\Out$ (by abusing notation slightly), with their obvious meanings. 

In this system model let $s=_{U\setminus G_H}s'$ mean that if we entered one of the states $s$ and $s'$ with a user not in $G_H$ the users and commands must be identical. 
If both states are entered by users in $G_H$ then the commands may be different. 
The output equality on states $s=_{O,G_L}s'$, shall refer to the observations the users $G_L$ can make at this state. 
Then, noninterference can be expressed as follows:
\[
\begin{array}{l}
\forall \pi. \forall\pi'.~ \pi[0]=_{U\setminus G_H}\pi'[0] ~~\W~~
\Big( \pi[0]\not=_{U\setminus G_H}\pi'[0] ~\wedge~ \\ 
\qquad\qquad \big(~\pi[0]\in H ~\wedge~ \LTLsquare(\LTLcircle\pi[0]=_{U\setminus G_H}\pi'[0]) ~\rightarrow~\LTLsquare(\LTLcircle\pi[0]=_{O,G_L}\pi'[0]) ~\big)\Big)
\end{array}
\]

\begin{theorem}\label{thm:noninterference_app}
There is a HyperLTL formula and an encoding $K(M)$ of state machines into Kripke structures such that for every state machine $M$, and groups of users $G_H$ and $G_L$ it holds $K(M)\models\varphi_{\mathit{NI}}(G_H,G_L)$ iff $G_H$ does not interfere with $G_L$ in $M$.  
\end{theorem}

\begin{proof}
The formula pattern $\pi[0]=_{U\setminus G_H}\pi'[0] ~\W~ \big(\pi[0]\not=_{U\setminus G_H}\!\pi'[0] ~\wedge~ \varphi \big)$ implies that $\varphi$ is applied exactly at the first position at which $\pi$ and $\pi'$ differ in their input (except, possibly, on the input of users in $G_H$). 
As, for a fixed path $\pi$, we quantify over all paths $\pi'$ the subformula $\varphi$ is hence applied to every position of every path~$\pi$. 

The subformula $\varphi$ then requires that, if from that position on $\pi$'s input differs from the input of $\pi'$ only in that it has an additional (secret) action by some user in $G_H$, both paths must look equivalent from the view point of the users in $G_L$. 
The interesting part here is the use of the next operator, as it enables the comparison of different positions of the traces. 
Thus, we compare path $\pi$ to all other paths that have one secret action less than itself. 
Hence, by transitivity of equivalence, we compare all paths $\pi$ to a version of itself that is stripped of all secret actions. 
\end{proof}

\subsection{Observational determinism}

\begin{theorem}
  There is an encoding of programs as defined in 
  \cite{HuismanWS/06/TLCharacterisationOfOD} into Kripke
  structures, such that \Hpctl{} can express observational determinism. 
\end{theorem}

\begin{proof}
Huisman et al.~\cite{HuismanWS/06/TLCharacterisationOfOD} define observational determinism over programs in a simple while language, which is very similar to the execution model of \cite{ZdancewicM03}. 
First, they define a special version of self-composition on their programs that allows both sides to move independently by introduces stuttering steps. 
Then, they encode programs into Kripke structures that do not only maintain the current state, but also remember the last state. 
The \ctlstr{} formula is then proved to precisely express observational determinism. 
Since \Hpctl{} subsumes \ctlstr{}, we can make use of the same encoding of programs into Kripke structures, but we leave out the self-composition operation on programs. 
In order to re-enable the correct synchronization of paths, however, we have to introduce stutter steps for single, not self-composed that is, programs. 
%

We prepend the \ctlstr{} formula with two path quantifiers over paths $\pi_1$ and $\pi_2$ and we let those propositions of the formula that referred to the two copies of the program now refer to the two paths, respectively. 
\end{proof}

\subsection{Declassification}

Likewise, HyperLTL can also express declassification properties in a programming language setting. 

\begin{corollary}
HyperLTL can express the declassification properties discussed in \cite{Balliu::Epis}. 
\end{corollary}

The result follows from the encoding~\cite{Balliu::Epis} of declassification
properties into epistemic temporal
logic~\cite{Fagin-book-95}
and the fact that HyperLTL subsumes epistemic temporal logic
(see Section~\ref{sect:relations}).

%% file: asynch.tex

\section{Epistemic Logics}
\label{app:asynchEpistemic}

We start with the proof that \Hpctl{} subsumes epistemic temporal logics under the assumption of synchronous time. 
That is, we assume the Kripke structure to have an atomic proposition that does nothing but switching its value in every step and that is observable by every agent. 
By this, the stutter-equivalent comparison collapses to a step-wise comparison. 

\begin{theorem}\label{thm:synchEpistemicApp}
Under synchronous time semantics, for every epistemic temporal logic formula $\psi$ and every set of agents $\Agts$, there is a HyperLTL formula $\varphi$ such that for all Kripke structures $K$ it holds $(K,\Agts)\models\psi$ iff $K\models\varphi$. 
\end{theorem}

\begin{proof}
We start by considering the extension of HyperLTL with the \emph{knowledge} operator with the semantics above. 
We prepend a universal path quantifier to the formula $\psi$ and apply a stepwise transformation to eliminate all knowledge operators.
(Note that the semantics relations $\pi,i\models\psi$ of linear time epistemic logics and and $\Pi\models_K\varphi$ of \Hpctl{} have different parameters, but it is plain how to join them into a single relation $\Pi,i\models_K\varphi$.)  
	
  Let $\varphi$ be a HyperLTL formula in NNF that possibly has knowledge operators.
  For brevity, we summarize the
  leading quantifiers of $\varphi$ with $\mathbf{Q}$, such that
  $\varphi=\mathbf{Q}.\varphi'$ where $\varphi'$ is quantifier-free.
  Let $t$ and $u$ be propositions that $\varphi$ does not refer to and that are not in the alphabet of the Kripke structure. 
	In case a knowledge operator ${\sf K}_A\psi'$ occurs in $\varphi'$ with positive polarity, translate $\varphi$ into the following HyperLTL formula: 
\[
\begin{array}{l}
\mathbf{Q}. \exists\pi.~\forall\pi'.~\forall\pi''.~\varphi'|_{{\sf K}_A\psi'\to u_{\pi}}\;\wedge\;
\big(
 (t_{\pi'}\U\LTLsquare \neg t_{\pi'}) ~\to~  \\\qquad  \LTLsquare(t_{\pi'}\to (\pi^*[0]\!=_A\!\pi''[0])) ~\to~ \LTLsquare(t_{\pi'}\wedge u_{\pi}\to [\psi']_{\pi''})\big)
\end{array}
\]
and, if the knowledge operator occurs negatively,
\[
\begin{array}{l}
\mathbf{Q}. \exists\pi.~\forall \pi'.~\exists\pi''.~\varphi'|_{\neg{\sf K}_A\psi'\to u_{\pi}}\;\wedge\;
\big( (t_{\pi'}\U\LTLsquare \neg t_{\pi'}) \to \\\qquad   \LTLsquare(t_{\pi'}\to (\pi^*[0]\!=_A\!\pi''[0]))  ~\wedge~ \LTLdiamond(t_{\pi'}\!\wedge u_{\pi}\wedge \neg [\psi']_{\pi''})\big)
\end{array}
\]
where $\varphi'|_{{\sf K}_P\psi\to u_{\pi}}$ denotes that in
$\varphi'$ one positive or one negative occurrence, respectively, of the knowledge
operator ${\sf K}_P\psi$ is replaced by proposition $u$, and where $\pi^*$ is the path on which the knowledge operator that we currently eliminate was applied on (which could be a different path to the first one, in case of nested knowledge operators). 
We repeat this transformation until no knowledge operators remain.

Paths $\pi$ and $\pi'$ are only used to carry the information about proposition $u$ and $t$, respectively. That is, we use these paths to quantify over a sequence of atomic propositions, the same way the QPTL quantification works. 
The sequence of the atomic proposition $u$ indicates at which positions of a given trace the knowledge operator needs to be evaluated. 
There can be multiple such sequences (e.g. for a formula $\LTLdiamond (\mathsf{K}_{A}\varphi'')$) that would make the formula satisfied, so the sequence is existentially quantified ($\exists\pi$). 
We need to fix one particular of these sequences of applications of the knowledge operator in advance, because we cannot go back in time to the initial state where the quantification knowledge operator happens. 

The sequence of atomic propositions $t$ is restricted to be true initially until it is globally false. 
In this way, we select a point in time until which the paths $\pi_n$ and $\pi''$ are required to be equal w.r.t. the observations of agent $A$ ($\LTLsquare(t_{\pi'}\to (\pi_n[0]=_A\pi''[0]))$). 
At this selected point in time, if also the knowledge operator is required to hold (indicated by $u$) we need to check that the subformula $\psi'$ of the knowledge operator is checked on the alternative path.
We relaxed this condition to simplify the formula. 
Instead of the point in the sequence where $t$ switches from true to false, we apply subformula $\psi'$ on \emph{all} positions before the that point ($\LTLsquare(t_{\pi'}\wedge u_{\pi}\to [\psi']_{\pi''})$). 
Since the sequences of $t$ are all-quantified, this does not affect the meaning of the formula. 
%
\end{proof}

Note that using past operator would simplify the encoding substantially, since we could more easily refer to the initial state where all paths in epistemic temporal logics branch. 

We now discuss the adaptations necessary to extend the encoding from the proof of Theorem~\ref{thm:synchEpistemicApp} to non-synchronous systems. 
For asynchronous systems it is natural to assume that they may stutter in their current state. 
In this discussion, we also assume that it is visible whether an asynchronous system has last performed a stuttering step, or whether it performed an action. (It may also be possible for a Kripke structure to perform an action that stays in a state.)
We call a Kripke structure with these properties an \emph{asynchronous Kripke structure}. 

The following proof of Theorem~\ref{thm:asynchEpistemic} is also a good example for the general treatment of asynchronous executions in \Hpctl{}.

\begin{theorem}
For every epistemic temporal logic formula $\psi$ and every set of agents $\Agts$, there is a HyperLTL formula $\varphi$ such that for all asynchronous Kripke structures $K$ it holds $(K,\Agts)\models\psi$ iff $K\models\varphi$. 
\end{theorem}

\begin{proof}
We assume that the atomic proposition $\text{stutter}$ describes that no action was performed in the last step. 

%

Since there always is the option to stutter in a stat, the number of paths we quantify over is larger. 
The idea of the proof is to restrict this enlarged set of paths to those that do not stutter forever ($\mathsf{progress}(\pi)=\LTLsquare\LTLdiamond\neg\text{stutter}_\pi$) and that synchronize correctly \big($\mathsf{synch}(\pi,\pi')=\LTLsquare\big( \pi[0]\not=_P\LTLcircle\pi[0] 
~~\Leftrightarrow~~ \pi'[0]\not=_P\LTLcircle\pi'[0]\big)$\big). 
By requiring the correct synchronization, we choose an alignment where changes in the observations happen in both paths at the same points in the sequences. 
That is, the remaining positions must be filled with stuttering steps. 


The encoding from the proof of Theorem~\ref{thm:synchEpistemicApp} is now modified as follows for positively occurring knowledge operators:
\[
\begin{array}{l}
\mathbf{Q}. \exists\pi. ~\varphi'|_{{\sf K}_A\psi'\to u_{\pi}}\;\wedge\;
\big(
\forall\pi'.~ (t_{\pi'}\U\LTLsquare \neg t_{\pi'}) ~\to~ \\\qquad 
\forall\pi''.~  \mathsf{progress}(\pi'') ~\wedge~ \mathsf{synch}(\pi,\pi'') ~\wedge~ \\\qquad  \LTLsquare(t_{\pi'}\to (\pi^*[0]\!=_A\!\pi''[0])) ~\to~ \LTLsquare(t_{\pi'}\wedge u_{\pi}\to [\psi']_{\pi''})\big)
\end{array}
\]
and, for negatively occurring knowledge operators:
\[
\begin{array}{l}
\mathbf{Q}. \exists\pi. ~\varphi'|_{{\sf K}_A\psi'\to u_{\pi}}\;\wedge\;
\big(\forall\pi'.~ (t_{\pi'}\U\LTLsquare \neg t_{\pi'}) \to \\\qquad 
\exists\pi''.~ \mathsf{progress}(\pi'') ~\wedge~ \mathsf{synch}(\pi,\pi'') ~\wedge~ \\\qquad
 \LTLsquare(t_{\pi'}\to (\pi^*[0]\!=_A\!\pi''[0]))  ~\wedge~ \LTLdiamond(t_{\pi'}\!\wedge u_{\pi}\wedge \neg [\psi']_{\pi''})\big)
\end{array}
\]

For a positively occurring knowledge operator $\mathcal{K}_A\psi'$, we prove that given two executions $\pi_1$ and $\pi_2$, if the sub-formula $\psi'$ must hold on $\pi_2$ at some position $i$ (by the semantics of the knowledge operator), then the formula above requires that on a stuttered version of $\pi_1$ the sub-formula is applied at state $\pi_2[i]$. 
We consider the prefixes of the two executions, $\pi_1[0,i]$ and $\pi_2[0,j]$ and assume they have the same traces with respect to agent $A$. 
For these, there are two stuttered versions $\pi$ and $\pi''$ (named to match the path variables in the formula above) that synchronize the positions at which they change their observations with respect to $A$ and pad the prefixes to the same length $k$ (without changing their final states, that is $\pi_2[i]=\pi''[k]$). 
There is also a labeling with $t_{\pi'}$ that ensures that the knowledge operator is evaluated until position $k$. 
Hence, $\psi'$ is applied on $\pi''$ at state $\pi_2[i]$. 

The other direction is straightforward, and the case of a negatively occurring knowledge operator follows similarly. 
The quantifiers do not occur inside the scope of temporal operators, and can thus easily be pulled to the front of the formula, resulting in a HyperLTL formula. 
\end{proof}

%% file: satisfiabilityAppendix.tex
\section{Satisfiability}
\label{app:satisfiability}

\begin{theorem}
  For {\Hpctl}, finite-state satisfiability is hard for $\Sigma^0_1$ and
  general satisfiability is hard for $\Sigma^1_1$.
\end{theorem}

\begin{proof}
  We give a reduction from the synthesis problem for LTL
  specifications in a distributed architecture consisting of two
  processes with disjoint sets of variables.  The synthesis problem
  consists on deciding whether there exist transition systems for the
  two processes with input variables $I_1$ and $I_2$, respectively,
  and output variables $O_1$ and $O_2$, respectively, such that the
  synchronous product of the two transition systems satisfies a given
  LTL formula $\varphi$. This problem is hard for $\Sigma^0_1$ if the
  transition systems are required to be finite, and hard for
  $\Sigma^1_1$ if infinite transition systems are allowed~(Theorems
  5.1.8 and 5.1.11 in \cite{Rosner/1992/Thesis}).

  To reduce the synthesis problem to {\Hpctl} satisfiability, we
  construct a {\Hpctl} formula $\psi$ as a conjunction $\psi=\psi_1
  \wedge \psi_2 \wedge \psi_3$.  The first conjunct ensures that
  $\varphi$ holds on all paths: $\psi_1 = \forall \pi. [\varphi]_\pi$,
	where $[\varphi]_\pi$ indicates that the atomic propositions in $\varphi$
get the index $\pi$. 
  The second conjunct ensures that every state of the model has a
  successor for every possible input: $\forall \pi . \LTLsquare
  \bigwedge_{I \subseteq I_1 \cup I_2} \exists \pi' \LTLcircle
  \bigwedge_{i \in I} i \ \bigwedge_{i \not\in I} \neg i$. The third
  conjunct ensures that the output in $O_1$ does not depend on $I_2$
  and the output in $O_2$ does not depend on $I_1$:
  $\psi_3=\forall\pi.\forall\pi'.~\big(\pi\!=_{I_1}\!\pi' \to
  \pi\!=_{O_1}\!\pi'\big)\wedge\big(\pi\!=_{I_2}\!\pi' \to
  \pi\!=_{O_2}\!\pi'\big)$.

The distributed synthesis problem has a (finite) solution iff the {\Hpctl} 
formula $\psi$ has a (finite) model.
\end{proof}

%% file: constructions.tex
\section{Model-checking Constructions for \Hp{}}
\label{sec:constructions}

\subsection{Self-composition construction}
\label{sec:selfcomp}

Self-composition is the technique that Barthe et al.~\cite{BartheDR04} adopt to verify noninterference policies.  It was generalized by Terauchi and Aiken~\cite{Terauchi+Aiken/05/SecureInformationFlowAsSafetyProblem} to verify observational determinism policies~\cite{Zdancewic:2005:policyDown,ZdancewicM03},
and by Clarkson and Schneider~\cite{ClarksonS10} to verify $k$-safety hyperproperties.
We  extend this technique to model-checking of {\Hp}.

\paragraph*{{\Buchi} automata.} {\Buchi} automata~\cite{Vardi:1994:infComp} are finite-state automata that accept strings of infinite length.
A {\Buchi} automaton is a tuple $(\Sigma,S,\Delta,S_0,F)$ where $\Sigma$ is an alphabet, $S$ is the set of states, $\Delta$ is the transition relation such that $\Delta\subseteq S\times\Sigma\times S$, $S_0$ is the set of initial states, and $F$ is the set of accepting states, where both $S_0\subseteq S$ and $F\subseteq S$.
A \emph{string} is a sequence of letters in $\Sigma$.
A path $s_0s_1\dots$ of a {\Buchi} automaton is \emph{over} a string $\alpha_1\alpha_2\dots$ if, for all $i\geq 0$, it holds that $(s_i,\alpha_{i+1},s_{i+1})\in\Delta$.
A string is \emph{recognized} by a {\Buchi} automaton if there exists a path $\pi$ over the string with some accepting states occurring infinitely often, in which case $\pi$ is an \emph{accepting path}.
The \emph{language} $\mathcal{L}(A)$ of an automaton $A$ is the set of strings that automaton accepts.
A {\Buchi} automaton can be derived~\cite{Clarke:1999:model-checking} from a Kripke structure, which   is a common mathematical model of interactive, state-based systems.

\paragraph*{Self composition.}
The \emph{$n$-fold self-composition} $A^n$ of {\Buchi} automaton $A$ is essentially the product of $A$ with itself, $n$ times. This construction is defined as follows:
\begin{definition}
{\Buchi} automaton $A^n$ is the \emph{$n$-fold self-composition} of {\Buchi} automaton $A$, where $A=(\Sigma,S,\Delta,S_0,F)$, if $A^n=(\Sigma^n,S^n,\Delta',S^n_0,F^n)$ and for all $s_1,s_2\in S^n$ and $\alpha\in\Sigma^n$ we have $(s_1,\alpha,s_2)\in\Delta' $ iff
for all $1\leq i\leq n$, it holds that $(\proj_i(s),\proj_i(\alpha),\proj_i(s'))\in\Delta$.
\end{definition}
Let $\zip$ denote the usual function that maps an $n$-tuple of sequences to a single sequence of $n$-tuples---for example, $\zip([1,2,3],[4,5,6]) = [(1,4), (2,5), (3,6)]$---and let $\unzip$ denote its inverse.
$A^n$ recognizes $zip(\pi_1, \dots, \pi_n)$ if $A$ recognizes each of $\pi_1,\dots,\pi_n$:
\begin{proposition}
$\mathcal{L}(A^n) = \{\zip(\pi_1,\dots,\pi_n) \mid \pi_1,\dots,\pi_n\in\mathcal{L}(A)\}$
\end{proposition}

\begin{proof}
By the construction of $A^n$.
\end{proof}

\subsection{Formula-to-automaton construction}
\label{sec:automataCons}

Given a {\Hp} formula $\forall\pi_1\dots\forall\pi_k\exists\pi_{k+1}\dots\exists\pi_{k+j}\psi$, we now show how to construct an automaton that accepts exactly the strings $w$ for which $\unzip(w)\models_\emptyset\psi$. 
Our construction extends standard methodologies for  LTL automata construction~\cite{GerthVardi:1995:OntheFlyLTL,Vardi:1996:LTL,Gastin:2001:FastLTLtoBuchi}.

\paragraph{1. Negation normal form.}
We begin by preprocessing $\psi$ to put it in a form more amenable to model checking.
The formula is rewritten to be in \emph{negation normal form} (NNF), meaning (i)  negation connectives are applied only to atomic propositions in $\psi$, (ii) the only connectives used in $\psi$ are \X{}, \U{}, \R{}, $\neg$, $\vee$ and $\wedge$.
We identify $\neg\neg\psi$ with $\psi$.

\paragraph{2. Construction.} 
We now construct a \emph{generalized {\Buchi} automaton}~\cite{CourcoubVardi:1992:MemEffAlgo} $A_{\psi}$ for $\psi$. 
A generalized {\Buchi} automaton is the same as a {\Buchi} automaton except that it has multiple sets of accepting states. 
That is, a generalized {\Buchi} automaton is a tuple $(\Sigma,S,\Delta,S_0,F)$ where $\Sigma$, $S$, $\Delta$ and $S_0$ are defined as for {\Buchi} automata, and $F=\setdef{F_i}{1\leq i\leq m \text{ and } F_i\subseteq S}$.
Each of the $F_i$ is an \emph{accepting set}.
A string is recognized by a generalized {\Buchi} automaton if there is a path over the string with at least one of the states in every accepting set occurring infinitely often.

To construct the states of $A_\psi$, we need some additional definitions. 
Define \emph{closure} $\cl(\psi)$ of $\psi$ to be the least set of subformulas of $\psi$ that is closed under the following rules:
\begin{itemize}
	\item if $\psi'\in \cl(\psi)$, then $\neg\psi'\in \cl(\psi)$.
	\item if $\psi_1\wedge \psi_2\in \cl(\psi)$ or $\psi_1\vee \psi_2\in \cl(\psi)$, then $\{\psi_1,\psi_2\}\subseteq \cl(\psi)$.
	\item if $\X \psi'\in \cl(\psi)$, then $\psi'\in \cl(\psi)$.
	\item if $\psi_1\U \psi_2\in \cl(\psi)$ or $\psi_1\R \psi_2\in \cl(\psi)$, then $\{\psi_1,\psi_2\}\subseteq \cl(\psi)$.
\end{itemize} 
And define $Mcs$ to be a \emph{maximal consistent set} with respect to $\cl(\psi)$ if $Mcs\subseteq\cl(\psi)$ and the following conditions hold:
\begin{itemize}
	\item $\psi'\in Mcs$ iff $\neg \psi'\not\in Mcs$.
\item if $\psi_1\wedge \psi_2\in \cl(\psi)$, then ($\psi_1\wedge \psi_2\in Mcs$ iff $\{\psi_1,\psi_2\} \subseteq Mcs$).
	\item if $\psi_1\vee \psi_2\in \cl(\psi)$, then ($\psi_1\vee \psi_2\in Mcs$ iff $\psi_1\in Mcs$ or $\psi_2\in Mcs$).
	\item if $\psi_1\U \psi_2\in Mcs$ then $\psi_1\in Mcs$ or $\psi_2\in Mcs$.
	\item if $\psi_1\R \psi_2\in Mcs$ then $\psi_2\in Mcs$.
\end{itemize} 

\noindent Define $\ms(\psi)$ to be the set of all maximal consistent sets with respect to $\psi$. 
The elements of $\ms(\psi)$ will be the states of $A_\psi$; hence each state is a set of formulas.
Intuitively, a state $s$ describes a set of trace tuples where each tuple satisfies all the formulas in $s$.  
There will be a transition from a state $s_1$ to a state $s_2$ iff every trace tuple described by $s_2$ is an \emph{immediate suffix} of some tuple described by $s_1$.
(Tuple $\Pi$ is an immediate suffix of $\Pi'$ iff $\Pi = \Pi'[1,\infty]$.)

Automaton $A_{\psi}=(\Sigma_{\psi},S_{\psi},\Delta_{\psi},\{\iota_{\psi}\},F_{\psi})$ is defined as follows:
\begin{itemize}
	\item The alphabet $\Sigma_{\psi}$ is $\mathcal{P}(\AP)^n$ where $\AP$ is the set of atomic propositions.  Each letter of the alphabet is, therefore, an $n$-tuple of sets of atomic propositions.
	\item The set $S_{\psi}$ of states is $\ms(\psi) \cup \{\iota_\psi\}$, where $\ms(\psi)$ is defined above and $\iota_{\psi}$ is a distinct initial state.
	\item The transition relation $\Delta_\psi$ contains $(s_1,\alpha,s_2)$, where  $\{s_1,s_2\} \subseteq S_{\psi}\setminus\{\iota_\psi\}$ and $\alpha\in\Sigma_{\psi}$, iff
		\begin{itemize}
			\item If $a_{\pi_i}\in s_2$ for some $1\leq i\leq n$, then $a\in \proj_i(\alpha)$.
				Likewise, if $\neg a_{\pi_i}\in s_2$, then $a\not\in \proj_i(\alpha)$.
			\item If $\X\psi'\in s_1$ then $\psi'\in s_2$. 
			\item If $\psi_1\U\psi_2\in s_1$ and $\psi_2\not\in s_1$ then $\psi_1\U\psi_2\in s_2$. 
			\item If $\psi_1\R\psi_2\in s_1$ and $\neg\psi_1\in s_1$ then $\psi_1\R\psi_2\in s_2$. 
		\end{itemize}
  And $\Delta_\psi$ contains $(\iota_{\psi},\alpha,s_2)$ iff $\psi\in s_2$ and $(\iota_{\psi},\alpha,s_2)$ is a transition permitted by the above rules for atomic propositions and their negations. 
    \item The set of initial states contains only $\iota_\psi$.
	\item The set $F_\psi$ of sets of accepting states contains one set $\setdef{s\in (S_{\psi}\setminus\{\iota_\psi\})}{\neg(\psi_1\U\psi_2)\in s \text{ or }\psi_2\in s}$ for each until formula $\psi_1\U\psi_2$ in $\cl(\psi)$.
\end{itemize}
The definition of $F_\psi$ guarantees that, for every until formula $\psi_1\U\psi_2$, eventually $\psi_2$ will hold.
	That is because the transition rules do not allow a transition from a state containing $\psi_1\U\psi_2$ to a state containing $\neg(\psi_1\U\psi_2)$ unless $\psi_2$ is already satisfied.

\paragraph{3. Degeneralization of {\Buchi} automata.} Finally, convert generalized {\Buchi} automaton $A_{\psi}$ to a ``plain'' {\Buchi} automaton. This conversion is entirely standard~\cite{GerthVardi:1995:OntheFlyLTL}, so we do not repeat it here.

\paragraph*{Correctness of the construction.}
Again without loss of generality, assume that the names of the traces are natural numbers;
then $\Pi$ is isomorphic to an $n$-tuple of traces.
Henceforth, we treat $\Pi$ as that tuple.
The following proposition states that $A_\psi$ is constructed such that it recognizes computation tuples that model $\psi$:
\begin{proposition}
 $\Pi\models_\emptyset\psi$ iff $\zip(\Pi)\in\mathcal{L}(A_{\psi})$.
\end{proposition}
\begin{proof}
($\Leftarrow$)
By the construction of $A_{\psi}$, the states with a transition from $\iota_{\psi}$ contain $\psi$. Hence by Lemma~\ref{lem:autoStrings} below, for all the strings $w$ such that $w=\zip(\Pi)$ in $\mathcal{L}(A_{\psi})$, it holds that $\Pi\models_\emptyset\psi$.
\par
($\Rightarrow$) 
Let $s_i=\{\psi'\in\cl(\psi)\mid \Pi[i,\infty]\models_\emptyset\psi'\}$ for all $i\geq 0$. Then by the definition, $s_i\in\ms(\psi)$. We show that $\iota_{\psi}s_0s_1\dots$ is an accepting path in $A_{\psi}$. By $\Pi\models_\emptyset\psi$ we have $\psi\in s_0$. By the construction of $A_{\psi}$, $(i_{\psi},\alpha_0,s_0)\in\Delta_{\psi}$ where $\alpha_0=\zip(\Pi)[0]$. The construction of the path inductively follows the construction of $A_{\psi}$, which respects the semantics of HyperLTL.
\end{proof}

\begin{lemma}
\label{lem:autoStrings}
Let $\iota_{\psi}s_0\dots$ be an accepting path in $A_{\psi}$ over the string $w=\alpha_0\alpha_1\dots$. Let $\Pi=\unzip(w)$. Then for all $i\geq 0$, it holds that $\psi'\in s_i$ iff $\Pi[i,\infty]\models_\emptyset\psi'$.
\end{lemma}
\begin{proof}
The proof proceeds by induction on the structure of $\psi'$:
\smallskip

\noindent\textbf{Base cases:}
\begin{enumerate}
	\item $\psi'=a_{\pi_r}$
	
	($\Rightarrow$) Assume that $a_{\pi_r}\in s_i$. By the construction of $A_{\psi}$, if $a_{\pi_r}\in s_i$ then $a\in \proj_r(\alpha_i)$ or equivalently $p\in \proj_r(\Pi)[i]$. By the semantics of HyperLTL, we have $\Pi[i,\infty]\models_\emptyset a_{\pi_r}$.
	
	($\Leftarrow$) Assume that $\Pi[i,\infty]\models_\emptyset a_{\pi_r}$.
	Then $a\in \proj_r(\Pi)[i]$, which is equivalent to $a\in \proj_r(\alpha_i)$. By the fact that states are maximal consistent sets, one of $a_{\pi_r}$ or $\neg a_{\pi_r}$ must appear in $s_i$.
	By the construction of $A_{\psi}$ and the fact that $a\in \proj_r(\alpha_i)$, we have $a_{\pi_r}\in s_i$.
	
\end{enumerate}

\noindent\textbf{Inductive cases:}
\begin{enumerate}
	\item $\psi'=\neg\psi''$
	
	($\Rightarrow$) Assume that $\neg\psi''\in s_i$, Then $\psi''\not\in s_i$. By induction hypothesis, $\Pi[i,\infty]\not\models_\emptyset\psi''$, or equivalently, $\Pi[i,\infty]\models_\emptyset\neg\psi''$. Hence, $\Pi[i,\infty]\models_\emptyset\psi'$.

	($\Leftarrow$) Similar to $\Rightarrow$.
	
	\item $\psi'=\psi_1\vee\psi_2$
	
	($\Rightarrow$) By the construction of $A_{\psi}$, if $\psi_1\vee\psi_2\in s_i$ then $\psi_1\in s_i$ or $\psi_2\in s_i$. By induction hypothesis, $\Pi[i,\infty]\models_\emptyset\psi_1$ or $\Pi[i,\infty]\models_\emptyset\psi_2$, which concludes $\Pi[i,\infty]\models_\emptyset\psi_1\vee\psi_2$.
	
	($\Leftarrow$) Similar to $\Rightarrow$.
	
	\item $\psi'=\X\psi''$ 
	
	($\Rightarrow$) Assume that $\psi'\in s_i$. By the  construction of $A_{\psi}$, $\psi''\in s_{i+1}$. By induction hypothesis, $\Pi[i+1,\infty]\models_\emptyset\psi''$, which concludes $\Pi[i,\infty]\models_\emptyset\X\psi''$.
	
	($\Leftarrow$) Similar to $\Rightarrow$ and the fact that always one of $\X\psi''$ or $\neg\X\psi''$ appears in a state.
	
	\item $\psi'=\psi_1\U\psi_2$
	
	($\Rightarrow$) Assume that $\psi_1\U\psi_2\in s_i$. By the construction of $A_{\psi}$ and the fact that the path is accepting, there is some $j\geq i$ such that $\psi_2\in s_j$. Let $j$ be the smallest index. By induction hypothesis, $\Pi[i,\infty]\models_\emptyset\psi_2$. By the construction of $A_{\psi}$, for all $i\leq k < j$, $\psi_1\in s_k$. Therefore by induction hypothesis, $\Pi[k,\infty]\models_\emptyset\psi_1$. which concludes $\Pi[k,\infty]\models_\emptyset\psi_1\U\psi_2$.

	($\Leftarrow$) Similar to $\Rightarrow$.
	
	\end{enumerate}

\end{proof}

%% file: prototypeDetails.tex
\section{Prototype Model Checker for {\Hp}}
\label{sec:prototypedetails}

The model-checking procedure for LTL works as follows:

\begin{enumerate}
\item Transform Kripke structure $K$ into its corresponding {\Buchi} automaton~\cite{Buchi62}, $A_K$.
This construction is standard~\cite{Clarke:1999:model-checking}. 
The language of $A_K$ is $\Traces(K)$.
\item Construct {\Buchi} automaton $A_{\neg\phi}$, whose language is the set of all traces that do not satisfy $\phi$.

\item Intersect $A_K$ and $A_{\neg\phi}$, yielding automaton $A_K \cap A_{\neg\phi}$. 
Its language contains all traces in $\Traces(K)$ that do not satisfy $\phi$.
This construction is standard~\cite{Clarke:1999:model-checking}.

\item Check whether the language of $A_K \cap A_{\neg\phi}$ is empty. 
If so, all traces $\Traces(K)$ satisfy $\phi$, hence we say $K$ satisfies $\phi$.
If not, then any element of the language is a counterexample showing that $K$ doesn't satisfy $\phi$.
\end{enumerate}

Our algorithm for model-checking {\Hp} adapts that LTL algorithm.
Without loss of generality, assume that the {\Hp} formula to be verified has the form $\forall\pi_{1..k}\exists\pi'_{1..j} \psi$, where $\forall\pi_{1..k}$ means $\forall\pi_1\dots\forall\pi_k$, and $\exists\pi'_{1..j}$ means $\exists\pi'_1\dots\exists\pi'_j$.
(Formulas of the form $\exists\pi_{1..k}\forall\pi'_{1..j}\psi$ can be verified by rewriting them as $\forall\pi_{1..k}\exists\pi'_{1..j}\neg\psi$.)
Let $n$ equal $k+j$.
Semantically, a model $\Pi$ of $\psi$ must be a set of named traces, where $|\Pi|=n$.
To determine whether a Kripke structure $K$ satisfies {\Hp} formula $\forall\pi_{1..k}\exists\pi'_{1..j} \psi$, our algorithm follows the same basic steps as the LTL algorithm:
\begin{enumerate}
\item Represent $K$ as a {\Buchi} automaton, $A_K$. 
Construct the $n$-fold product of $A_K$ with itself---that is, $A_K \times A_K \times \cdots \times A_K$, where ``$A_K$'' occurs $n$ times.
This construction is straightforward and formalized in Appendix~\ref{sec:constructions}.
Denote the resulting automaton as $A^n_K$.
If $\pi_1,\ldots\pi_n$ are all traces of $\Traces(K)$, then $\zip(\pi_1,\ldots\pi_n)$ is a word in the language of $A^n_K$.%

\item Construct {\Buchi} automaton ${A_{\psi}}$.
Its language is the set of all words $w$ such that $\unzip(w)=\Pi$ and $\Pi \models_{\emptyset} \psi$---that is, the tuples $\Pi$ of traces that satisfy $\psi$.
This construction is a generalization of the corresponding LTL construction.
It is formalized in Appendix~\ref{sec:constructions}.

\item Intersect $A^n_K$ and ${A_{\psi}}$, yielding automaton ${A^n_K\cap A_{\psi}}$.  
Its language is essentially the tuples of traces in $\Traces(K)$ that satisfy $\psi$.  

\item \label{step:complement} Check whether $\mathcal{L}({((A^n_K\cap A_{\psi})|_k)^C\cap A^k_K})$ is empty, where
(i) $A^C$ denotes the \emph{complement} of an automaton $A$,
(complement constructions are well-known---e.g.,~\cite{Vardi:1996:LTL}---so we do not formalize one here), 
and (ii) $A|_k$ denotes the same automaton as $A$, but with every transition label (which is an $n$-tuple of propositions) \emph{projected} to only its first $k$ elements.
That is, if $\mathcal{L}(A)$ contains words of the form $\zip(\pi_1,\ldots\pi_n)$, then $\mathcal{L}(A|_k)$ contains  words of the form $\zip(\pi_1,\ldots\pi_k)$.
Projection erases the final $j$ traces from each letter of a word, leaving only the initial $k$ traces.
Thus a word is in the projected language iff there exists some extension of the word in the original language. 

If $\mathcal{L}({((A^n_K\cap A_{\psi})|_k)^C\cap A^k_K})$ is empty, then it holds that $$\emptyset\models_{\Traces(K)}\forall\pi_{1..k}\exists\pi'_{1..j} \psi.$$
If not, then any element of the language is a counterexample showing that $\emptyset\not\models_{\Traces(K)}\forall\pi_{1..k}\exists\pi'_{1..j} \psi$.

\end{enumerate}

The final step of the above algorithm is a significant departure from the LTL algorithm.  
Intuitively, it works because projection introduces an existential quantifier, thus enabling verification of formulas with a quantifier alternation.
The following theorem states the correctness of our algorithm:
\begin{theorem}\label{thm:modelcheckingcorrect}
Let $\phi$ be {\Hp} formula $\forall\pi_{1..k}\exists\pi'_{1..j}\psi$, and let $n = k+j$.  
Let $K$ be a Kripke structure.
Then $\phi$ holds of $K$ iff $\mathcal{L}({((A^n_K\cap A_{\psi})|_k)^C\cap A^k_K})$ is empty.
\end{theorem}
\begin{proof}
($\Rightarrow$, by contrapositive)
We seek a countermodel showing that $\forall\pi_{1..k}\exists\pi'_{1..j}\psi$ doesn't hold of $K$.
For that countermodel to exist, 
\begin{equation}\label{eq:mc1}
\begin{split}
&\text{there must exist a $k$-tuple $\Pi_k$ : for all $j$-tuples $\Pi_j$ :} \\ &\text{if $\mathit{set}(\Pi_k \cdot \Pi_j) \subseteq \Traces(K)$ then $\Pi_k \cdot \Pi_j \models_\emptyset \neg\psi$,}
\end{split}
\end{equation}
where $\mathit{set}(\Pi)$ denotes the set containing the same elements as tuple $\Pi$.
To find that countermodel $\Pi_k$, consider $\mathcal{L}(A^n_K\cap A_{\psi})$.
If that language is empty, then 
\begin{equation}\label{eq:mc2}
\begin{split}
&\text{for all $k$-tuples $\Pi_k$ and for all $j$-tuples $\Pi_j$ :} \\ &\text{if $\mathit{set}(\Pi_k \cdot \Pi_j) \subseteq \Traces(K)$ then $\Pi_k \cdot \Pi_j \models_\emptyset \neg\psi$.}
\end{split}
\end{equation}
That's almost what we want, except that $\Pi_k$ is universally quantified in~\eqref{eq:mc2} rather than existentially quantified as in~\eqref{eq:mc1}.
So we introduce projection and complementation to relax the universal quantification to existential.
First, note that language $\mathcal{L}((A^n_K\cap A_{\psi})|_k)$ contains all $\zip(\Pi_k)$ for which there exists a $\Pi_j$ such that $\mathit{set}(\Pi_k \cdot \Pi_j) \subseteq \Traces(K)$ and $\Pi_k \cdot \Pi_j \models_\emptyset \psi$.
So if there exists a $\Pi_k^*$ such that $\zip(\Pi_k^*) \not\in \mathcal{L}((A^n_K\cap A_{\psi})|_k)$, then for all  $\Pi_j$, if $\mathit{set}(\Pi_k \cdot \Pi_j) \subseteq \Traces(K)$ then $\Pi_k \cdot \Pi_j \models \neg\psi$.
That $\Pi_k^*$ would be exactly the countermodel we seek according to~\eqref{eq:mc1}.
To find such a $\Pi_k^*$, it suffices to determine whether $\mathcal{L}((A^n_K\cap A_{\psi})|_k) \subset \mathcal{L}(A^k_K)$, because any element that strictly separates those sets would satisfy the requirements to be a $\Pi_k^*$.
By simple set theory, $X \subset Y$ iff $X^C \cap Y$ is not empty.
Therefore, if $\mathcal{L}({((A^n_K\cap A_{\psi})|_k)^C\cap A^k_K})$ is not empty, then a countermodel $\Pi_k^*$ exists.

($\Leftarrow$) The same argument suffices: if $\mathcal{L}({((A^n_K\cap A_{\psi})|_k)^C\cap A^k_K})$ is empty, then no countermodel can exist.
\end{proof}

\paragraph*{Formulas without quantifier alternation.} 
Define HyperLTL$_1$ to be the fragment of {\Hp} that contains formulas with no alternation of quantifiers.
HyperLTL$_1$ can be verified more efficiently than {\Hp}.
To verify $\forall\pi_{1..n}\psi$,
it suffices to check whether $A^n_M \cap A_{\neg\psi}$ is non-empty.
This is essentially the self-composition construction, as used in previous work~\cite{BartheDR04,Terauchi+Aiken/05/SecureInformationFlowAsSafetyProblem,ClarksonS10}.

\paragraph*{Complexity.}
The most expensive computation in model checking {\Hp} is the {\Buchi} automaton complementation.
Safra's construction~\cite{Safra:1988:complexity}, which can be used to implement complementation, has complexity of $2^{O(m \log m)}$, where $m$ is the number of the states of the original automaton.
In step~\ref{step:complement} of our algorithm, the automaton being complemented has $O(|K|^n\cdot 2^{O(|\psi|)})$ states, where $|K|$ is the number of states of $K$, and $|\psi|$ is the length of $\psi$.
Combined with the final intersection and non-emptiness check, the complexity of model checking is $O\left(|K|^n \cdot 2^{O(|K|^{n+1} \cdot O(|\psi|) \cdot 2^{O(|\psi|)})}\right)$.
Our prototype is therefore exponential in the size of the program (i.e., Kripke structure)  and doubly exponential in the size of the formula. 
That complexity is worse than the complexity of model-checking LTL, which is polynomial in the size of the program and exponential in the size of the formula~\cite{Vardi:1996:LTL}.
Perhaps the complexity of model-checking {\Hp} could be reduced in future work, or perhaps it's simply the price we pay to have a general-purpose logic of hyperproperties.

The worst case time complexity of epistemic logic with perfect recall is similar to {\Hp}, as it is exponential to the size of the system, and doubly exponential in the size of the formula~\cite{Cohen:2010:NonElementarySpeedUp}.
The worst-case space complexity of single-quantifier-alternation HyperLTL formulas (i.e., many of the security policies of interest) is NLOGSPACE, whereas epistemic with perfect recall is worse: PSPACE-hard~\cite{Meyden:1999:MCKnowledge}.
So there is hardly any theoretical advantage of epistemic over HyperLTL.